\documentclass[12pt]{scrartcl}

\usepackage{latexsym,amsfonts,amsmath,epsfig,tabularx,amsthm,dsfont,mathrsfs}
\usepackage{graphicx}
\usepackage{enumerate,hyperref}
\usepackage{authblk}

\theoremstyle{definition}

\newtheorem{theorem}{Theorem} 
\newtheorem{remark}[theorem]{Remark}
\newtheorem{algorithm}{Algorithm}
\newtheorem{lemma}[theorem]{Lemma}
\newtheorem{proposition}[theorem]{Proposition}
\newtheorem{assumption}[theorem]{Assumption}
\newtheorem{definition}[theorem]{Definition}
\newtheorem{example}[theorem]{Example}
\newtheorem{corollary}[theorem]{Corollary}

\newcommand{\widebar}[1]{\mbox{\kern1.5pt\hbox{\vbox{\hrule height 0.6pt \kern0.35ex
        \hbox{\kern-0.15em \ensuremath{#1 }\kern0.0em}}}}\kern-0.1pt}

\newcommand{\E}{\mathbb{E}}
\renewcommand{\P}{\mathbb{P}}

\renewcommand{\d}{{\mathrm{d}}}
\newcommand{\dint}{{\mathrm{d}}}

\DeclareMathOperator{\Var}{Var}

\DeclareMathOperator{\diag}{diag}
\DeclareMathOperator{\Cov}{Cov}
\DeclareMathOperator{\Corr}{Corr}
\newcommand{\mc}{\mathcal}

\DeclareMathOperator{\argmin}{argmin}

\newcommand{\abs}[1]{\left\vert #1 \right\vert}
\newcommand{\norm}[1]{\left\Vert #1 \right\Vert}

\newlength{\fixboxwidth}
\usepackage{color}
\definecolor{darkred}{RGB}{139,0,0}
\definecolor{darkgreen}{RGB}{0,100,0}
\definecolor{darkmagenta}{RGB}{139,0,139}
\definecolor{darkpurple}{RGB}{110,0,180}
\definecolor{darkblue}{RGB}{40,0,200}
\definecolor{darkorange}{RGB}{255,140,0}

\title{On a Metropolis--Hastings importance sampling estimator} 

\author{Daniel Rudolf\thanks{Email: daniel.rudolf@uni-goettingen.de} }
\affil{Institute for Mathematical Stochastics, University of Goettingen, Goldschmidtstr. 7, 37077 G\"ottingen, Germany}
\author{Bj\"orn Sprungk\thanks{Email: bjoern.sprungk@math.tu-freiberg.de}}
\affil{Faculty of Mathematics and Computer Science, Technische Universit\"at Bergakademie Freiberg, 09596 Freiberg, Germany}

\begin{document}
\maketitle

\begin{abstract}
A classical approach for 
approximating expectations of functions w.r.t. 
partially known distributions
is to compute the average of function values along a trajectory of a
Metropolis--Hastings (MH) Markov chain. 
A key part in the MH algorithm is a suitable acceptance/rejection
of a proposed state, which ensures the correct stationary distribution 
of the resulting Mar\-kov chain.
However, the rejection of proposals causes highly correlated samples.
In particular, when a state is rejected it is not taken any further into account.
In contrast to that we consider
a MH importance sampling estimator which explicitly incorporates 
all proposed states generated by the MH algorithm.
The estimator satisfies a strong law of large numbers as well as a central limit theorem, and, in addition to that,
we provide an explicit mean squared error bound. 
Remarkably, the asymptotic variance of 
the MH importance sampling estimator does not involve any correlation term 
in contrast to its classical counterpart. 
Moreover, although the analyzed estimator 
uses the same amount of information as the 
classical 
MH estimator, it 
can outperform the latter in scenarios of moderate dimensions as indicated by numerical experiments.
\end{abstract}

%
%

\section{Introduction}
\emph{Motivation.}
A fundamental task in computational science and statistics is the computation of expectations 
w.r.t. a partially unknown probability measure $\mu$ on a measurable space $(G,\mathcal{G})$ 
determined by
\begin{equation} \label{equ:mu}
	 \frac{\d\mu}{\d\mu_0}(x) = \frac{\rho(x)}{Z}, \quad x\in G,
\end{equation}
where $\mu_0$ denotes a $\sigma$-finite reference measure on $G$ and where the normalizing 
constant $Z= \int_{G} \rho(x)\, \mu_0(\d x)\in(0,\infty)$ is typically unknown.
Thus, given a function $f\colon G\to \mathbb R$ 
the goal is to compute $\mathbb{E}_\mu(f)  = \int_{G} f(x) \mu(\d x)$ 
only by using evaluations of $f$ and $\rho$.
Here, a plain Monte Carlo estimator for 
the approximation of $\mathbb{E}_\mu(f)$ based on independent $\mu$-distributed random variables is, in general, 
infeasible due to the unknown normalizing constant $Z$ 
and the fact that we only have access to function evaluations of $\rho$.
However, a possible and very common approach is the construction of a Markov chain for approximate sampling w.r.t. $\mu$.
In particular, the well-known Metropolis--Hastings (MH) algorithm provides a general scheme for simulating 
a Markov chain $(X_n)_{n\in \mathbb{N}}$ with stationary distribution $\mu$. 
Under appropriate assumptions the distribution of 
$X_n$ of such a \emph{MH Markov chain} converges 
to $\mu$ and the \emph{classical MCMC estimator}  for $\mathbb{E}_\mu(f)$ is 
then given by the sample average 
\begin{equation} \label{eq: classical_MCMC_est}
	S_n(f) = \frac{1}{n} \sum_{k=1}^n f(X_k).
\end{equation}
The statistical efficiency of $S_n(f)$ highly depends on the
autocorrelation of the time series $(f(X_n))_{n\in \mathbb{N}}$. 
In particular,
a large 
autocorrelation diminishes the efficiency of $S_n(f)$.
An essential part in the MH algorithm is the acceptance/rejection step:
Given $X_n=x$, 
a sample $y$ of $Y_{n+1}\sim P(x,\cdot)$ is drawn, where 
$P$ denotes a \emph{proposal transition kernel}. 
But only with a certain probability this $y$ is accepted as the next state, that is $X_{n+1}:=y$, 
and otherwise it is rejected, such that $X_{n+1}:=x$. 
This indicates that a potential reason for a high autocorrelation 
is the rejection of proposed states. 
Hence, the question arises whether it is possible to derive a more efficient estimator for $\mathbb{E}_\mu(f)$ based on the potentially less correlated time series $(f(Y_n))_{n\in\mathbb N}$ determined by the sample of proposals $Y_n$.

\emph{Main Result.}
In this paper we consider and analyze a modification of the classical estimator from
\eqref{eq: classical_MCMC_est}
of the form
\[
  A_n(f) = \frac{\sum_{k=1}^n w(X_k,Y_k) f(Y_k)}{\sum_{k=1}^n w(X_k,Y_k)},
 \]
which we call \emph{MH importance sampling estimator}. 
The (importance) weight $w$ is chosen in such a way that we obtain a consistent estimator.
More detailed, we set $w(x,y) := \frac{\d \mu_0}{\d P(x,\cdot)}(y)\cdot \rho(y)$ assuming the existence of the density $\frac{\d \mu_0}{\d P(x,\cdot)}$ for each $x\in G$.
The appeal of the modified estimator is that it is still 
based on the MH algorithm and needs no additional
function evaluations of $\rho$ and $f$, 
while, after appropriate tuning in scenarios of moderate dimensions, it can outperform the classical estimator as we illustrate in a few numerical examples in Section~\ref{sec_num_exam}. 
Moreover, it can be seen and studied as an importance 
sampling corrected MCMC estimator, 
or as an importance sampling  estimator using an 
underlying MH Markov chain for providing 
the importance distributions.
In this paper we have chosen the first point of view and exploit the fact that the
\emph{augmented MH Markov chain} $(X_n,Y_n)_{n\in\mathbb{N}}$ inherits 
several desirable\footnote{Surprisingly, the augmented MH Markov chain
is in general not reversible but still has a stationary distribution.} properties of the 
original MH Markov chain $(X_n)_{n\in \mathbb{N}}$ such as Harris recurrence, 
see Lemma~\ref{lem: K_aug_prop}.
By using those properties we prove the following results for the estimator $A_n$:
\begin{itemize}
\item  Theorem~\ref{thm: SLLN}: A strong law of large numbers (SLLN), i.e., for functions $f \in L^1(\mu)$ we have almost surely $A_n(f) \to \mathbb{E}_\mu(f)$ as $n\to \infty$;
\item Theorem~\ref{thm: CLT}: A central limit theorem (CLT), that is, for any $f \in L^2(\mu)$ the scaled error $\sqrt{n}(A_n(f)-\mathbb{E}_\mu(f))$ converges in distribution to a mean-zero normal distribution $\mathcal{N}(0,\sigma^2_A(f))$ with \emph{asymptotic variance} $\sigma^2_A(f)$ given by
\[
 \sigma_A^2(f):= \int_G\int_G (f(y)-\mathbb{E}_\mu(f))^2\frac{\dint \mu}{\dint P(x,\cdot)}(y) \mu(\dint y)\mu(\dint x);
\]
\item Theorem~\ref{thm: mse}: An estimate of the mean squared error $ \mathbb{E}\abs{A_n(f)-\mathbb{E}_\mu(f)}^2$ for bounded functions $f\colon G \to \mathbb R$.
\end{itemize}
Here, we denote by $L^p(\mu)$, $p\in[1,\infty)$ the Lebesgue space of functions $f\colon G\to \mathbb R$ which are $p$-integrable w.r.t.~$\mu$.
It is remarkable that in the asymptotic variance $\sigma_A^2(f)$ of the CLT there is no covariance or correlation term. 
However, there appears the density of $\mu$ w.r.t. $P(x,\cdot)$ 
which quantifies the difference of the employed importance distribution given by the proposal transition kernel $P(x,\cdot)$ and the desired distribution $\mu$. 

\emph{Related literature.} Importance sampling is a well-established 
technique for approximating expectations, see \cite{CaMouRy05,Ow13} 
for textbook introductions, which has recently attracted considerable attention 
in terms of theory and application, see for example \cite{AgPaSaSt17,ChDi15,Hi10,SCh15}.
In particular, its combination with Markov chain Monte Carlo methods is exploited by several authors.
For example, Botev et al. \cite{BoLETu13} use the MH algorithm in order 
to approximately sample from the minimum variance importance distribution.
Vihola et al. \cite{ViHeFr16} consider general importance sampling estimators 
based on an underlying Mar\-kov chain and Martino et al. \cite{MaElLuCo16} 
propose a hierarchical approach where a mixture importance distribution 
close to $\mu$ is constructed based on the (accepted) samples $X_k$ in the MH algorithm.
Schuster and Klebanov \cite{SchKle18} follow a similar idea to the latter, 
but rather use the proposals $Y_k$ of the MH algorithm and their asymptotic 
distribution as the importance distribution.
Indeed, the idea of using all proposed states generated in the MH algorithm for 
estimating expectations such as $\mathbb E_\mu(f)$ is not new.
For instance, Frenkel suggests in \cite{Fr04,Fr06} an approximation scheme 
which recycles the rejected states in a MH algorithm. 
In the work of Delmas and Jourdain \cite{DeJo09} this method is used in a 
control variate variance reduction approach and it is analyzed in a general framework. 
It turns out that for the Barker-algorithm the method is indeed beneficial, whereas 
for the MH-algorithm this is not necessarily the case. 
In particular, an estimator similar to $A_n(f)$ as above but for sampling from normalized 
densities was already introduced by Casella and Robert \cite{CaRo96}.
However, besides some 
numerical examples it was not further studied in \cite{CaRo96} whereas their main focus, 
variance reduction of sampling methods by Rao-Blackwellization, got extended by \cite{AtPe05,DoRo11}.
In particular, the theoretical results of Douc and Robert \cite{DoRo11} 
provide variance reduction guarantees for their MH based estimator while 
keeping the additional computation cost under control.
In contrast to that, using the estimator $A_n$ 
does not increase the number of function evaluations, but we also do not provide a guarantee of improvement.

\emph{Outline.}
First, we provide some basic preliminaries on Markov chains and the corresponding 
classical MCMC estimator $S_n$. 
In Section~\ref{sec: novel_est} we introduce 
the MH importance sampling estimator, study 
properties of the aforementioned 
augmented MH Markov chain $(X_n,Y_n)_{n\in\mathbb{N}}$ 
and state the main results.
In Section~\ref{sec_num_exam} we compare the classical MCMC estimator $S_n$ with $A_n$  
numerically in two representative examples and draw some conclusions 
in Section \ref{sec:concl}.

\section{Preliminaries on Markov chain Monte Carlo} \label{sec:basics_MCMC}
Let $(\Omega,\mathcal{F},\mathbb{P})$ be a probability space. 
The random variables considered throughout the paper (mainly)
map from this probability space to a measurable space $(G,\mathcal{G})$.
A \emph{(time-homogeneous) Markov chain} is a sequence of random variables $(X_n)_{n\in \mathbb{N}}$
which satisfy for any $A\in \mathcal{G}$ and any $n\in \mathbb{N}$ that $\mathbb P$-almost surely
\[
 \mathbb{P}(X_{n+1}\in A\mid X_1,\dots,X_n) = K(X_n,A),
\]
where $K\colon G\times \mathcal{G} \to [0,1]$ denotes a 
\emph{transition kernel}, i.e., $K(x,\cdot)$ 
is a probability measure for any $x\in G$ and 
the mapping $x\mapsto K(x,A)$ is measurable for any $A\in \mathcal{G}$.
Our focus is on Markov chains 
designed
for approximate
sampling of the distribution $\mu$. 
Such Markov chains typically have $\mu$ as 
their \emph{stationary distribution}, i.e., 
their transition kernels $K$ satisfy 
$\mu K = \mu$, where $\mu K(A) := \int_G K(x,A)\ \mu(\d x)$ for any $A\in\mathcal G$.

\subsection{The Metropolis--Hastings algorithm}
\label{sec: MH_alg}

Let $P\colon G\times \mathcal{G} \to [0,1]$ be a \emph{proposal transition kernel} satisfying the following structural assumption.
\begin{assumption} \label{assum:density_0}
For any $x\in G$ the proposal $P(x,\cdot)$ possesses a density $p(x,\cdot)$ w.r.t.~$\mu_0$ and 
for any $y\in G$ assume
\begin{equation*}
	\rho(y) > 0 \quad \Longrightarrow \quad p(x,y) >0 \quad \forall x\in G.
\end{equation*}
\end{assumption}
This condition has some useful implications, see Proposition \ref{propo:irreducible}. 
Moreover, for example for $G \subseteq \mathbb R^d$, $\mathcal G = \mathcal B(G)$ 
and $\mu_0$ being the Lebesgue measure, any Gaussian proposal, such as a Gaussian- or Langevin-random walk, satisfies it. 
Assumption \ref{assum:density_0} allows us to define the finite ``acceptance ratio'' $r(x,y)$ for the MH algorithm for any $x,y\in G$ according to \cite[Section~2]{Ti98} by
\[
 r(x,y) := 
 \begin{cases}
  \frac{\rho(y) p(y,x)}{\rho(x) p(x,y) } & \rho(x) p(x,y) > 0,\\
  1		& \text{otherwise}.
 \end{cases}
\] 
Then, the \emph{MH algorithm}, which provides a realization of a Markov chain
$(X_n)_{n\in \mathbb{N}}$, works as follows:
\begin{algorithm}  \label{alg: Metro_Hast}
 Assume that $X_n=x$, then the next state $X_{n+1}$ 
 is generated by the following steps: \begin{enumerate}
  \item Draw $Y_n\sim P(x,\cdot)$ and $U\sim \text{Unif}[0,1]$ independently, 
  call the result $y$ and $u$, respectively.
  \item \label{it: acc_step}
	Set 
	$
	 \alpha(x,y) := \min\left\{1,r(x,y)\right\}.
	$
  \item 
  Accept $y$ with probability $\alpha(x,y)$, 
 that is, if $u<\alpha(x,y)$, then set
  $X_{n+1}=y$, otherwise set $X_{n+1}=x$.
 \end{enumerate}
\end{algorithm}

The Markov chain generated by the MH algorithm is called \emph{MH Markov chain},
and its transition kernel, which we also call \emph{MH (transition) kernel}, is given by
\begin{align}\label{equ:MetroKern}
 K(x,A) & := \int_A \alpha(x,y) P(x,\dint y) 
+ \mathbf{1}_A(x) \int_G \alpha^c(x,y) P(x,\dint y), \quad A\in \mathcal{G},
\end{align}
where $\alpha^c(x,y):= 1-\alpha(x,y)$.
It is well-known that the transition kernel $K$ in \eqref{equ:MetroKern}
is reversible w.r.t. $\mu$, that is, $K(x,\dint y)\mu(\dint x) = K(y,\dint x)\mu(\dint y)$. 
In particular, this implies that $\mu$ is a stationary distribution of $K$.

\subsection{Strong law of large numbers, central limit theorem and mean squared error bound}

For convergence, in particular the strong law of large numbers, 
we need the concepts of 
$\phi$-irreducibility and Harris recurrence: 
Given a $\sigma$-finite measure $\phi$
on $(G,\mathcal{G})$, a Markov chain $(X_n)_{n\in \mathbb{N}}$ is 
\emph{$\phi$-irreducible}
if for each $A\in \mathcal{G}$ with $\phi(A)>0$ and each $x\in G$
there exists an $n=n(x,A)\in \mathbb{N}$ such that
$
 \mathbb{P}(X_n\in A\mid X_1 = x)>0.
$
Furthermore, a Markov chain $(X_n)_{n\in \mathbb{N}}$ is \emph{Harris recurrent}
if it is $\phi$-irreducible and satisfies for each $A\in \mathcal{G}$ with
$\phi(A)>0$ that for any $x\in G$
\[
\mathbb{P}(X_n\in A \text{ infinitely often}\mid X_1=x)=1.
\] 
It is proven in \cite[Corollary~2]{Ti94} 
that $\mu$-irreducibility of a MH Markov 
chain $(X_n)_{n\in \mathbb{N}}$ implies Harris recurrence.
Moreover, it is known that Assumption~\ref{assum:density_0} ensures 
$\mu$-irreducibility and, thus, Harris recurrence:
\begin{proposition}[{\cite[Lemma~1.1]{MeTw96}}]\label{propo:irreducible}
Given Assumption \ref{assum:density_0} 
the Markov chain $(X_n)_{n\in \mathbb N}$ 
realized by the MH algorithm is $\mu$-irreducible.
\end{proposition}

We recall the SLLN of the classical MCMC estimator 
$S_n(f)$ given in \eqref{eq: classical_MCMC_est}
based on the concept of Harris recurrence.
\begin{theorem}[SLLN for $S_n$, {\cite[Theorem~17.0.1]{MeTw93}}] 
\label{thm: SLLN_gen}
 Let $(X_n)_{n\in \mathbb{N}}$ be a Harris recurrent Markov chain with stationary
 distribution $\mu$ on $G$ and let $f\in L^1(\mu)$.
 Then,
 \begin{equation*}
  S_n(f)\ 
 \xrightarrow[n \to \infty]{\text{a.s.}}\ 
  \mathbb{E}_\mu(f),
 \end{equation*}
for any \emph{initial distribution}, i.e., any distribution of $X_1$.
\end{theorem}
This theorem justifies that the classical MCMC method based on the MH algorithm yields a consistent estimator.
Moreover, for $S_n(f)$ also a central limit theorem can be shown. 
Deriving a CLT is an important issue in studying MCMC and a lot of 
conditions which imply a CLT are known, for an overview we refer to the survey paper \cite{Jo04} and the references therein.
We require some further terminology.
Let 
$\mathrm K\colon L^2(\mu) \to L^2(\mu)$ be the \emph{transition operator} 
associated to the transition kernel $K$ of a Markov 
chain $(X_n)_{n\in \mathbb{N}}$ given by
\[
(\mathrm K f)(x) := \int_G f(y) K(x,\dint y), \quad f \in L^2(\mu).
\]
For $n\geq 2$ and $f\in L^2(\mu)$ we have
\[
 \mathrm K^n f(x) = \int_G f(y) K^n(x,\dint y),
\]
where $K^n$ is the $n$-step transition kernel, which is recursively defined by
\[
 K^n(x,A) = \int_G K(y,A) K^{n-1}(x,\dint y),\quad A\in \mathcal{G}.
\]
Note that the transition operator recovers the transition kernel, namely, for $n\geq1$ we have 
\[
	({\mathrm K}^n \mathbf{1}_A)(x)=K^n(x,A),
	\qquad
	x\in G, \ A\in\mathcal G.
\]
We also need the concept of the asymptotic variance:
Let $(X^*_n)_{n\in \mathbb{N}}$ denote a Markov chain 
with transition kernel $K$ starting at stationarity, i.e., the stationary distribution $\mu$ 
is also the initial one.
Then, for $f\in L^2(\mu)$ the \emph{asymptotic variance} 
of the classical MCMC estimator $S_n(f)$ for $\mathbb E_\mu(f)$
is given by
\[
 \sigma_{S}^2(f) := 
 \lim_{n\to \infty} n \cdot \Var\left(\frac{1}{n}\sum_{k=1}^n f(X_k^*)\right)
\]
whenever the limit exists. 
 One can easily see that the asymptotic variance admits the following
 representation in terms of the autocorrelation of the 
 time series $(f(X^*_n))_{n\in\mathbb{N}}$. Namely,
 \begin{equation}\label{equ:asymp_var_cor}
 \sigma_{S}^2(f) = 
 \Var_\mu(f)\left(1+2\sum_{k=1}^\infty \Corr(f(X^*_1),f(X^*_{1+k})) \right),
 \end{equation}
where $\Var_\mu(f) :=\mathbb{E}_\mu(f-\mathbb{E}_\mu(f))^2$ denotes the variance of $f$ w.r.t.~$\mu$ and 
$\Corr(\cdot,\cdot)$ the correlation between random variables.

\begin{theorem}[CLT for $S_n$]  \label{thm: gen_CLT}
  Let $(X_n)_{n\in \mathbb{N}}$ be a Harris recurrent 
  Markov chain with transition kernel $K$ and stationary distribution $\mu$. 
  For $f\in L^2(\mu)$, if either
  \begin{enumerate}
   \item \label{en: maxwell_woodroofe}
    $\sum_{k=1}^{\infty} k^{-3/2} \left(\mathbb{E}_{\mu}[\sum_{j=0}^{k-1} \mathrm K^j(f-\mathbb{E}_\mu(f))]^2\right)^{1/2}<\infty$ or
   \item \label{en: kipnis}
   $K$ is reversible w.r.t. $\mu$ and $\sigma_{S}^2(f)<\infty$,
  \end{enumerate}
  then we have for any initial distribution
  \[
   \sqrt{n}(S_n(f)-\mathbb{E}_\mu(f))\ 
   \xrightarrow[n\to\infty]{\mathcal{D}}\
   \mathcal{N}(0,\sigma_{S}^2(f))
  \]
 with $\sigma_{S}^2(f)$ as in \eqref{equ:asymp_var_cor}.
 \end{theorem}
 The theorem is justified by the following arguments.
 First, by \cite[Proposition~17.1.6]{MeTw96} it is sufficient to have a CLT when
 the initial distribution is a stationary one. In that case the Markov chain is an ergodic stationary process.
 Under condition~\ref{en: maxwell_woodroofe}., where no reversibility is necessary, the statement follows then by arguments derived in the introduction of
\cite{MaWo00}.
Under condition~\ref{en: kipnis}. the statement follows based on \cite[Corollary~1.5]{KiVa86}. 
Although MH Markov chains are $\mu$-reversible by construction, we encounter in the following a non-reversible Markov chain and derive a CLT by 
verifying \ref{en: maxwell_woodroofe}.

The SLLN and the CLT only contain asymptotic statements, but one might be interested
in explicit error bounds.
For $f\in L^2(\mu)$ the \emph{mean squared error} of the classical MCMC estimator $S_n(f)$ is given by
$
 \mathbb{E}\abs{S_n(f)-\mathbb{E}_\mu(f)}^2
$.
Depending on different convergence properties of the underlying 
Markov chain different error bounds are known, see for example \cite{JoOl10,LaMiNi13,LaNi11,Ru09,Ru10,Ru12}.
In particular, there is a relation between 
the asymptotic variance $\sigma_{S}^2(f)$ 
and the mean squared error of $S_n$: If $X_1 \sim \mu$, then
\[
 \lim_{n\to \infty} n \, \cdot \mathbb{E}\abs{S_n(f)-\mathbb{E}_\mu(f)}^2 = \sigma^2_{S}(f),
\]
and some of the error bounds have the same asymptotic behavior, see \cite{LaMiNi13} and also \cite{Ru12}.

\section{The MH importance sampling estimator}
\label{sec: novel_est}
The CLT for the MCMC estimator $S_n(f)$ shows that its statistical efficiency determined by the 
asymptotic variance $\sigma_S^2(f)$ is diminished by a large autocorrelation of 
$(f(X^*_n))_{n\in\mathbb{N}}$ or $(f(X_n))_{n\in\mathbb{N}}$, respectively.
A reason for a large autocorrelation is the rejection of proposed states.
In particular, the sequence of proposed states $(f(Y_n))_{n\in\mathbb{N}}$ 
is potentially less correlated than the MH Markov chain itself, since no rejection is involved.
For example, if a proposal kernel $P$ 
on $G = \mathbb R^d$ is absolutely 
continuous w.r.t.~the Lebesgue measure,
and $X_n \sim \mu$, then we have
\begin{align*}
	0 = \mathbb P(Y_{n+1} = Y_n)
	& \leq 
	\mathbb P(X_{n+1} = X_n)
	 = \int_G \alpha^c(x,y)\, P(x, \d y)\ \mu(\d x).
\end{align*}
Thus, one may ask whether it is beneficial, in terms of a higher statistical efficiency, 
to consider an estimator based on $(f(Y_n))_{n\in\mathbb{N}}$ rather 
than $(f(X_n))_{n\in\mathbb{N}}$.
Such an estimator might be of the form
\[
A_n(f) = \frac{ \sum_{k=1}^n w_k f(Y_k)}{\sum_{k=1}^n w_k}
\]
with suitable weights $w_k$.
The reason for the latter is the fact that $Y_n\sim P(X_n,\cdot)$ does 
not follow the distribution $\mu$. 
In fact, even if $X_n\sim \mu$, then $Y_n \sim \mu P$, hence, we need to apply an 
importance sampling correction 
in order to obtain a consistent estimator $A_n(f)$.
To this end, Assumption \ref{assum:density_0} ensures the existence of:
\begin{equation}\label{equ:rho_bar}
	\bar\rho(x,y) := Z\ \frac{\d \mu}{\d P(x,\cdot)}(y) \qquad \forall x,y\in G.
\end{equation}
Indeed, by the fact that $p(x,y)=0$ implies $\rho(y)=0$ (Assumption~\ref{assum:density_0})
we have
\[
	\bar\rho(x,y)
	=
	\begin{cases}
	\rho(y)/p(x,y),& \rho(y) > 0,\\
	0, & \rho(y) = 0.
	\end{cases}
\]
Moreover, the acceptance ratio $r(x,y)$ can be expressed only in terms of $\bar\rho$:
\[
 r(x,y) = 
 \begin{cases}
  \frac{\bar\rho(y,x)}{\bar\rho(x,y)} & \bar\rho(x,y) > 0,\\
  1		& \text{otherwise}.
 \end{cases}
\] 
As it turns out,
$\bar \rho$ provides the correct weights $w_k$ for an estimator $A_n(f)$, as indicated by the next result.
\begin{proposition}\label{prop:ew}
   Let Assumption \ref{assum:density_0} be satisfied.
 Then, for any $f\in L^1(\mu)$, we have
 \[
 \mathbb{E}_\mu(f)  = 
 \frac{\int_{G} \int_{G} f(y) {\bar \rho}(x,y) P(x,\d y) \mu(\d x) }
      { \int_{G} \int_{G} {\bar \rho}(x,y) P(x,\d y)  \mu (\d x)}
 \]
with $\bar{\rho}$ as in \eqref{equ:rho_bar}.
 \end{proposition}
\begin{proof}
We have
\begin{align*}
 \mathbb{E}_\mu(f) & 
 = \int_{G} f(y) \frac{\bar{\rho}(x,y)}{Z} P(x,\d y) =\int_{G} \int_{G} f(y) \frac{\bar{\rho}(x,y)}{Z} P(x,\d y) \mu(\d x) \\
 & = \frac{\int_{G} \int_{G} f(y) \bar{\rho}(x,y) P(x,\d y) \mu(\d x) }
      { \int_{G} \int_{G} \bar{\rho}(x,y) P(x,\d y) \mu(\d x)}, 
\end{align*}
where the last equality follows from
\begin{align*}
 Z
 & = \int_{G} \rho(y) \mu_0(\d y) 
   = \int_{G} \frac{\d \mu_0}{\d P(x,\cdot)}(y) \rho(y) P(x,\d y) \\
 & = \int_{G} \int_{G} \frac{\d \mu_0}{\d P(x,\cdot)}(y) \rho(y) P(x,\d y) \mu(\d x)\\
 &  = \int_{G} \int_{G} \bar{\rho}(x,y) P(x,\d y) \mu(\d x).
\end{align*}
\end{proof}

Proposition~\ref{prop:ew} motivates the following estimator.
\begin{definition}
Let Assumption~\ref{assum:density_0} be satisfied 
and let $(X_n)_{n\in\mathbb N}$ be a MH Markov chain, 
where $(Y_n)_{n\in\mathbb N}$ 
denotes the corresponding proposal sequence.
Then, given $f\in L^1(\mu)$,
the \emph{MH importance sampling estimator} for $\mathbb{E}_\mu(f)$ is
\begin{align}\label{equ:A_n}
	 A_n(f) := \frac{\sum_{k=1}^n \bar{\rho}(X_k,Y_k) f(Y_k)}{ \sum_{k=1}^n \bar{\rho}(X_k,Y_k)}
\end{align}
with $\bar{\rho}$ defined in \eqref{equ:rho_bar}.
\end{definition}

\begin{remark}
The dependence on $\rho$ in $A_n$ is explicitly given within $\bar{\rho}$, 
whereas the dependence on $\rho$ of the classical estimator $S_n$ realized with the MH algorithm is rather implicit. 
Namely, it appears only in the acceptance probability of the MH algorithm. 
However, in many situations the computational cost for function 
evaluations of $\rho$ are much larger than for function evaluations
of $f$, such that it seems counterintuitive to use the information of the value 
of $\rho$ at the proposed state, 
which was expensive to compute, 
not any further. 
\end{remark}

\begin{remark}
The estimator $A_n(f)$ is related to 
self-nor\-malizing importance sampling estimators for $\mathbb{E}_\mu(f)$ of the form
\[
 \frac{\sum_{k=1}^n w_k f(\xi_k)}{\sum_{k=1}^n w_k},
\]
where $(\xi_k)_{k\in\mathbb{N}}$ is an arbitrary sequence of random variables $\xi_k \sim \phi_k$
and where $w_k=\frac{\dint \mu_0}{\dint \phi_k}(\xi_k)\rho(\xi_k)$ are the 
corresponding 
importance weights.
For $(\xi_k)_{k\in\mathbb{N}} = (Y_k)_{k\in\mathbb{N}}$ 
being the proposal sequence 
in the MH algorithm for realizing a $\mu$-reversible 
Markov chain $(X_k)_{k\in\mathbb{N}}$, 
we recover $A_n(f)$ 
with $\phi_k = P(X_k,\cdot)$.
In other words, $A_n(f)$ can be viewed as an importance sampling estimator where the importance
distributions $\phi_k$ are determined by a MH Markov chain.
\end{remark}

\begin{remark}
Related to the previous remark we highlight a recent approach similar but slightly different to ours.
Namely, the authors of \cite{SchKle18} propose 
and study a self-normalizing importance sampler 
where the importance distribution is $\phi_k = \mu P$, i.e., 
the stationary distribution of the proposal sequence in the MH algorithm. 
Moreover, we remark that the particular form of the estimator $A_n(f)$ in the 
case of already normalized weights appeared in \cite[Section 5]{CaRo96}, but without any further analysis.
Since self-normalizing is rather inevitable in practice, we 
continue studying $A_n(f)$ as in \eqref{equ:A_n}.
\end{remark}

\subsection{The augmented MH Markov chain and its properties} \label{sec:augMC}
In order to analyze the MH importance sampling estimator $A_n$ 
we consider the \emph{augmented MH Markov chain} $(X_n,Y_n)_{n\in\mathbb{N}}$ 
on $G\times G$ consisting of the original MH Markov chain 
$(X_n)_{n\in\mathbb{N}}$ and the 
associated sequence of proposals $(Y_n)_{n\in\mathbb{N}}$. 
The transition kernel $K_\text{aug}$ of the augmented MH Markov chain is given by
\[
 K_\text{aug}\left( (x,y), \d u  \d v \right) 
 := \delta_{y}(\d u)P(y,\d v)\alpha(x,y)
  + \delta_{x}(\dint u) P(x,\d v) \alpha^c(x,y)
\]
for $x,y \in G$, where $\delta_{z}$ denotes the Dirac-measure at $z\in G$.
Now we derive a useful representation of $K_{\text{aug}}$ and the MH kernel $K$, 
which simplify several arguments.
To this end, we define the probability measure 
\begin{equation}
\label{equ:nu}
 \nu(\d x\d y):= P(x,\d y)\mu(\d x)
\end{equation}
on $(G\times G,\mathcal{G}\otimes\mathcal{G})$ and let $L^2(\nu)$ 
be the space of functions $g\colon G\times G\to \mathbb{R}$
which satisfy
\[
 \Vert g \Vert_{\nu} := \left(\int_{G \times G} \vert g(x,y) \vert^2 \nu(\d x \d y) \right)^{1/2}<\infty.
\]
By $K_{\text{aug}}$ the transition operator ${\mathrm K}_{\text{aug}}\colon L^2(\nu)\to L^2(\nu)$ is induced.
Furthermore, for a given proposal transition kernel $P$ we define a linear operator ${\mathrm{\widehat P}} \colon L^2(\nu) \to L^2(\mu)$ by
\[
 ({\mathrm{\widehat P}} g)(x) := \int_G g(x,y)P(x,\d y).
\]
It is easily seen that its adjoint operator ${\mathrm{\widehat P}}^* \colon L^2(\mu)\to L^2(\nu)$
is given by 
\[({\mathrm{\widehat P}}^* f)(x,y)=f(x),\] 
i.e., 
$\langle {\mathrm{\widehat P}}g,f\rangle_\mu = \langle g,{\mathrm{\widehat P}}^*f\rangle_\nu$, where
$\langle\cdot,\cdot \rangle_\mu$ and $\langle\cdot,\cdot \rangle_\nu$
denote the inner products in $L^2(\mu)$ and $L^2(\nu)$, respectively. 
Let $H$ be the transition kernel on $G\times G$ given by
\begin{equation*}
 	H((x,y), \d u \d v)  := \alpha(x,y) \delta_{(y,x)}(\d u \d v)
	+  \alpha^c(x,y) \delta_{(x,y)}(\d u \d v)
\end{equation*}
and let ${\mathrm H} \colon L^2(\nu) \to L^2(\nu)$ denote the associated transition operator.
The following properties are useful for the subsequent analysis.

\begin{lemma}
\label{lem: K_aug_prop}
With the above notation  
we have that
\begin{enumerate}
  \item\label{it: H_self_adj}
  ${\mathrm H}$ is self-adjoint and $\norm{{\mathrm H}}_{L^2(\nu)\to L^2(\nu)}=1$;

  \item\label{it: projection} 
  $\mathrm{\widehat P^* \widehat P}\colon L^2(\nu)\to L^2(\nu)$ 
  is a projection and \[\Vert{\mathrm{\widehat P}}\Vert_{L^2(\nu)\to L^2(\mu)}
        =\Vert{\mathrm{\widehat P}^*}\Vert_{L^2(\mu)\to L^2(\nu)} =1;\]
  \item\label{it: aux_repr} 
  ${\mathrm K}={\mathrm{\widehat{P} H \widehat{P}^*}}$ and   
  ${\mathrm K}_{\text{aug}}={\mathrm{ H \widehat{P}^*\widehat{P}}}$;
  \item\label{it: stat_K_aug} 
 $\nu$ given in \eqref{equ:nu} is a stationary distribution of $K_{\text{aug}}$;
  \item\label{it: aux_repr_n}
  ${\mathrm K}^n_{\text{aug}}={\mathrm{ H \widehat P^*}} {\mathrm K}^{n-1} {\mathrm{\widehat{P}}}$ 
	  and ${\mathrm K}^n= {\mathrm{\widehat P}} {\mathrm K}_{\text{aug}}^{n-1} \mathrm{ H \widehat P^*}$ 
	  for $n\geq 2$.
  \end{enumerate}
\end{lemma}
\begin{proof}
 \textbf{ To \ref{it: H_self_adj}.:} 
 Let $g_1, g_2 \in L^2(\nu)$. 
  Then, by the choice of $\alpha(x,y)$ we have
 $
  \alpha(x,y)\nu(\d x \d y) = \alpha(y,x) \nu(\d y \d x),
 $
 and self-adjointness follows from
 \begin{align*}
  \langle \mathrm H g_1,g_2 \rangle_\nu 
  & = \int_{G\times G} (\alpha(x,y)g_1(y,x)
  +\alpha^c(x,y)g_1(x,y))g_2(x,y)\nu(\d x\d y)\\
  & = \int_{G\times G} g_1(y,x)g_2(x,y) \alpha(x,y)\nu(\d x\d y) + \int_{G\times G} \alpha^c(x,y)g_1(x,y)g_2(x,y) \nu(\d x \d y)\\
  & = \int_{G\times G} g_1(x,y)g_2(y,x) \alpha(y,x)\nu(\d x\d y) + \int_{G\times G} \alpha^c(x,y)g_1(x,y)g_2(x,y) \nu(\d x \d y)\\
  & = \langle g_1,\mathrm Hg_2 \rangle_\nu.    
 \end{align*}
 Since ${\mathrm H}$ is induced by the transition kernel $H$ 
 the operator norm is one.\\
 \textbf{To \ref{it: projection}.:} It is easily seen that 
 $\mathrm{\widehat P^* \widehat P}$ is a projection. 
 Moreover, it is well-known that the norm of an 
 operator and its adjoint coincide, which yields the statement in combination with  
 \[
 1 = \norm{\mathrm{\widehat P^* \widehat P}}_{L^2(\nu)\to L^2(\nu)}
 =   \Vert{\mathrm{\widehat P}}\Vert_{L^2(\nu)\to L^2(\mu)}.
 \] 
 \textbf{To \ref{it: aux_repr}.:} 
 The representations can be verified by a straightforward calculation.\\
 \textbf{To \ref{it: stat_K_aug}.:} For any $A,B\in \mathcal{G}$ we have
 \begin{align*}
  \nu K_{\text{aug}}(A\times B)& = 
  \int_{G^2} ({\mathrm{ H\widehat P^*\widehat P}} \mathbf1_{A\times B})(x,y) P(x,\dint y)\mu(\dint x)\\
  & = \int_G ({\mathrm{\widehat P}}{\mathrm{ H\widehat P^*\widehat P}} \mathbf1_{A\times B})(x)\mu(\dint x)\\
  &  = \int_G (\mathrm{ K \widehat P}\mathbf1_{A\times B})(x) \mu(\dint x)\\
  & = \int_G (\mathrm{\widehat P}\mathbf1_{A\times B})(x) \mu(\dint x) = \nu(A\times B),
 \end{align*}
where the last-but-one equality follows from the fact that $\mu$ is a stationary
distribution of $K$. Since the Cartesian products $A\times B$ 
provide a generating system of 
$\mathcal{G}\otimes\mathcal{G}$ the result follows
by the uniqueness theorem of probability measures.\\
\textbf{To \ref{it: aux_repr_n}.:} 
These representations are a direct consequence of \ref{it: aux_repr}.
 \end{proof}

Note that statement \ref{it: aux_repr_n} of Lemma~\ref{lem: K_aug_prop}  
yields for $n\geq1$ and $g\in L^2(\nu)$ that
\begin{equation} \label{eq: for_K_aug_low_est}
\begin{split}
      (K^n_{\text{aug}}\,g)(x,y) & = \alpha(x,y) \int_{G^2} g(u,v) P(u,\dint v) K^{n-1}(y,\dint u) \\
   & \qquad \phantom{a}+ \alpha^c(x,y) \int_{G^2} g(u,v) P(u,\dint v) K^{n-1}(x,\dint u).
\end{split}
\end{equation}

\begin{remark}
\label{rem: K_aug_not_rev}
In general, the transition kernel $K_\text{aug}$ is not reversible w.r.t. $\nu$. Since reversibility is equivalent to self-adjointness of the Markov operator this can be seen by the fact that $K_{{\mathrm aug}}^* = \widehat{P}^* \widehat{P} H$, which does not necessarily coincides with $K_{\mathrm aug}$. For convenience of the reader we also provide a simple example which illustrates the non-reversibility.
 Consider a finite state space $G=\{1,2\}$ equipped with the counting measure $\mu_0$ 
 with $\rho(i)=1/2$ and $P(i,j)=1/2$ for all $i,j\in G$ such that $\alpha(i,j)=1$.
 Then the transition matrix $K_\text{aug}$ is given by
 \[
  {K}_{\mathrm aug}((i,j),(k,\ell)) 
  = \frac{\delta_{j}(\{k\})}{2}
 \]
 for any $i,j,k,\ell\in G$.
 Here reversibility is equivalent to
$  K_\text{aug}((i,j),(k,\ell)) = K_\text{aug}((k,\ell),(i,j)) $ for all $i,j,k,\ell\in G$, which is not satisfied for
 $i=j=\ell=1$ and $k=2$.
\end{remark}

Now, using Lemma~\ref{lem: K_aug_prop}
we show that stability properties 
of the MH kernel $K$ pass over to $K_\text{aug}$. 
The proof of the following result is adapted from \cite[Lemma 24]{ViHeFr16}.

\begin{lemma}
\label{theo:Harris}
Assume that $\phi$ is a $\sigma$-finite measure on $(G,\mathcal{G})$ 
and let $K$ denote the MH kernel as in \eqref{equ:MetroKern}. 
\begin{itemize}
\item
If $K$ is $\phi$-irreducible, then $K_\text{aug}$ is $\phi_P$-irreducible on $G\times G$,
where the $\sigma$-finite measure $\phi_P$ is given by
\mbox{$\phi_P(\d x \d y) := P(x, \d y) \phi(\d x) $}.

\item
If $K$ is Harris recurrent (w.r.t.~$\phi$), then $K_\text{aug}$ is also Harris recurrent 
(w.r.t.~$\phi_P$).
\end{itemize}
\end{lemma}
\begin{proof}
For $A \in \mathcal{G}\otimes \mathcal{G}$ and $x\in G$ define
\begin{align*}
	A_2(x) &:= \{y\in G: (x,y)\in A\} \in \mc G,\\
	A_1 &:=  \{x\in G: A_2(x) \neq \emptyset \}\in \mc G,
\end{align*}
so that $A_2(x)$ is the slice of $A$ for fixed first component $x$ 
and $A_1$ is the ``projection'' of the set $A$ on the first component space.
For $\varepsilon>0$ let
\[
	A_1(\varepsilon) := \{x\in G: P(x, A_2(x)) > \varepsilon\}. 
\]
By the use of \eqref{eq: for_K_aug_low_est}
we prove the irreducibility statement: 
Assume that $A\in\mc G\otimes \mc G$ with $\phi_P(A) > 0$.
Then, $\phi(A_1) > 0$, since otherwise
\[
	\phi_P(A) = \int_A P(x,\d y) \phi(\d x) = \int_{A_1} P(x, A_2(x)) \, \phi(\d x) 
\]
is zero.
By the same argument, one obtains that there exists an $\varepsilon >0$ 
such that $\phi(A_1(\varepsilon)) > 0$, since otherwise 
\[
	\phi_P(A) = \int_{\bigcup_{\varepsilon > 0 } A_1(\varepsilon)} P(x, A(x)) \, \phi(\d x)
\]
is zero.
Because of the $\phi$-irreducibility of $K$, we have
 for $x,y \in G$ that there exist $n_x, n_y \in \mathbb{N}$ such that 
$K^{n_x}(x, A_1(\varepsilon)) > 0$ and $K^{n_y}(y, A_1(\varepsilon)) > 0$. 
Hence, if $\alpha(x,y) > 0$, then
\begin{align*}
	 K_\text{aug}^{n_y+1}((x,y), A)
	& \overset{\eqref{eq: for_K_aug_low_est}}{\geq} \alpha(x,y) \int_A P(u,\d v) K^{n_y}(y, \d u) \\
	&= \alpha(x,y)\int_{A_1} P(u, A_2(u)) K^{n_y}(y, \d u)\\
	& \geq \alpha(x,y)\int_{A_1(\varepsilon)} P(u, A_2(u)) K^{n_y}(y, \d u)\\ 
	&\geq \alpha(x,y)\,\varepsilon\, K^{n_y}(y, A_1(\varepsilon)) > 0.
\end{align*}
Otherwise, if $\alpha^c(x,y)=1$, we obtain analogously
\begin{align*}
	K_\text{aug}^{n_x+1}((x,y), A) & \geq \alpha^c(x,y) \int_A P(u,\d v) K^{n_x}(x, \d u)
	 \geq \varepsilon\, K^{n_x}(x, A_1(\varepsilon)) > 0.
\end{align*}
In other words, for $(x,y)\in G\times G$ we find an $n\in \mathbb{N}$ 
(depending on $\alpha(x,y)$) such that $K_\text{aug}^{n}((x,y), A)>0$, 
which proves the $\phi_P$-irreducibility.

We turn to the Harris recurrence: Let $K$ be Harris recurrent w.r.t.~$\phi$ 
and let $\phi_P(A) > 0$.
As above, we can conclude that there exists an $\varepsilon >0$ 
such that $\phi(A_1(\varepsilon))>0$.
Furthermore, for the augmented 
Markov chain $(X_n,Y_n )_{n\in\mathbb{N}}$ 
with transition kernel $K_\text{aug}$ 
we have
\begin{align*}
	\P\left( (X_n,Y_n) \in A \right)
	&=
	\P\left( Y_n \in A_2(X_n) \right)
	=
	P(X_n, A_2(X_n)).
\end{align*}
By $\phi(A_1(\varepsilon))>0$ and the fact that $(X_n)_{n\in\mathbb{N}}$ is 
Harris recurrent 
w.r.t.~$\phi$, with probability one there are infinitely many distinct times 
$(\tau_k)_{k\in\mathbb{N}}$, such that $X_{\tau_k} \in A_1(\varepsilon)$ for any $k\in\mathbb{N}$.
Hence
\begin{align*}
	\P\left( \sum_{n=1}^\infty \mathbf{1}_A(X_n,Y_n) = \infty \right)
 &=\P\left( \sum_{n=1}^\infty \mathbf{1}_{A_2(X_n)}(Y_n)= \infty\right) \geq \P\left( \sum_{k=1}^\infty \mathbf{1}_{A_2(X_{\tau_k})}(Y_{\tau_k})= \infty\right).
\end{align*}
Note that by construction $\mathbf{1}_{A_2(X_{\tau_k})}(Y_{\tau_k})$ 
are Bernoulli random variables with success probability of at least $\varepsilon$.
Moreover, they are conditionally independent given $(X_{\tau_k})_{k\in \mathbb{N}}$.
Hence,
\[
	\P\left( \sum_{k=1}^\infty \mathbf{1}_{A_2(X_{\tau_k})}(Y_{\tau_k}) 
	= \infty \,\bigg|\, (X_{\tau_k})_{k\in \mathbb{N}} \right)
	= 1
	\qquad
	\P\text{-a.s.}
\]
yields
\begin{align*}
	& \P\left( \sum_{k=1}^\infty \mathbf{1}_{A_2(X_{\tau_k})}(Y_{\tau_k}) = \infty\right) 
	 = \E\left[ \P\left( \sum_{k=1}^\infty \mathbf{1}_{A_2(X_{\tau_k})}(Y_{\tau_k})= \infty \,\bigg|\, (X_{\tau_k})_{k\in \mathbb{N}} \right) \right]
	 = 1,
\end{align*}
which shows that the augmented MH Markov chain is Harris recurrent.
\end{proof}

\begin{remark}
Another consequence of Lemma~\ref{lem: K_aug_prop} interesting on its own
is that also geometric ergodicity is inherited by the augmented MH Markov chain.
However, since this fact is not relevant for the remainder of the paper, 
we postpone the discussion of geometric ergodicity 
and its inheritance to Appendix \ref{sec:geomErgodic}.
\end{remark}

\subsection{Strong law of large numbers and central limit theorem}

A consistency statement in form of a SLLN 
of the MH importance sampling estimator defined in \eqref{equ:A_n} is stated 
and proven in the following. 
A key argument in the proofs is the inheritance of 
Harris recurrence of $(X_n)_{n\in\mathbb{N}}$
to the augmented MH Markov chain $(X_n,Y_n)_{n\in\mathbb{N}}$. 

\begin{theorem}
\label{thm: SLLN}
Let Assumption~\ref{assum:density_0} be satisfied.
Then, for any initial distribution and any $f\in L^1(\mu)$ 
we have 
\begin{align} \label{equ:A_consistent}
	A_n(f) = \frac{ \frac{1}{n} \sum_{k=1}^n \bar{\rho}(X_k,Y_k) f(Y_k)}
	{ \frac{1}{n} \sum_{k=1}^n \bar{\rho}(X_k,Y_k)}\
	\xrightarrow[n\to\infty]{\text{a.s.}}\ \mathbb{E}_\mu(f).
\end{align}
\end{theorem}
\begin{proof}
Assumption~\ref{assum:density_0} implies $\mu$-irreducibility and Harris recurrence 
of the MH Markov chain $(X_n)_{n\in\mathbb{N}}$ due to Proposition \ref{propo:irreducible}.
This yields, due to Lemma~\ref{theo:Harris}, 
that also the transition kernel $K_\text{aug}$ is 
Harris recurrent.
Hence, by Theorem~\ref{thm: SLLN_gen} we have for each $h \in L^1(\nu)$ 
that 
\[
	\frac 1n \sum_{k=1}^n h(X_k,Y_k) 
	\xrightarrow[n\to\infty]{\text{a.s.}}
	\mathbb{E}_{\nu}(h).
\]
Define $h_1(x,y) := \bar{\rho}(x,y) f(y)$ and $h_2(x,y) := \bar{\rho}(x,y)$.
Since $\E_\nu(h_2) = Z < \infty$ and $\E_\nu(h_1) = \E_{\mu}(f)\cdot Z < \infty$, 
we have $h_1,h_2 \in L^1(\nu)$ and, thus, the numerator and denominator on the left-hand side of \eqref{equ:A_consistent} converge a.s.~to $\E_{\nu}(h_1)$ and $\E_{\nu}(h_2)$.
The assertion follows then by the continuous mapping theorem 
and $\E_{\nu}(h_1) / \E_{\nu}(h_2) = \E_{\mu}(f)$.
\end{proof}

The next goal is to derive a CLT, which provides 
a way to quantify the asymptotic behavior of $A_n$. 
Since
the augmented 
Markov chain $(X_n,Y_n)_{n\in\mathbb{N}}$ is, in general, not reversible w.r.t. $\nu$, we aim to use condition \ref{en: maxwell_woodroofe} of Theorem~\ref{thm: gen_CLT}.

\begin{theorem}  \label{thm: CLT}
Let Assumption~\ref{assum:density_0} be satisfied
 and assume for $f\in L^1(\mu)$ that
 \[
  \sigma^2_{A}(f) := 
  \int_G \int_G (f(y)-\mathbb{E}_\mu(f))^2 \frac{\dint \mu}{\dint P(x,\cdot)}(y) \, 
  \mu(\dint y) \mu(\dint x)
 \]
is finite.
 Then, for any initial distribution, we have 
 \[
  \sqrt{n}(A_n(f)-\mathbb{E}_\mu(f))\ 
  \xrightarrow[n\to\infty]{\mathcal D}\
  \mathcal{N}(0,\sigma^2_{A}(f)).
 \]
\end{theorem}
\begin{proof}
 We frequently use the identity 
 \begin{equation} \label{eq: freq_used_fact}
   \int_G g(x,y)\bar{\rho}(x,y) P(x,\dint y)
   =
   Z \int_G g(x,y) \mu(\dint y),
 \end{equation}
 for any $x\in G$ and any $g\colon G^2 \to \mathbb{R}$ for which one of the two 
 integrals exist.
  Define the centered version of $f$ by $f_c(y):= f(y)-\mathbb{E}_\mu(f)$ and set
  $h_3(x,y):= \bar{\rho}(x,y)f_c(y)$ for $x,y\in G$. 
  Note that $\mathbb{E}_\nu(h_3)=0$ and $h_3\in L^2(\nu)$, since 
 \begin{align*}
\mathbb E_\nu(h_3^2) & = \int_G \int_G f_c(y)^2 \bar{\rho}(x,y)^2 P(x,\dint y) \mu(\dint x) \\
   & \overset{\eqref{eq: freq_used_fact}}{=}
    Z \int_G \int_G f_c(y)^2 \bar{\rho}(x,y) \mu(\dint y) \mu(\dint x)\\
& =
  Z^2 \int_G \int_G f_c(y)^2 \frac{\dint \mu}{\dint P(x,\cdot)}(y) \mu(\dint y)\mu(\dint x)\\
&  = Z^2 \sigma^2_{A}(f)<\infty.
 \end{align*}
With  the representation \eqref{eq: for_K_aug_low_est} one obtains for any $k\geq 2$ that
 \begin{align*}
  \mathrm K_\text{aug}^k h_3(x,y)
  & = \int_{G\times G} \bar{\rho}(u,v)f_c(v) K_\text{aug}^k(x,y,\dint u\,\dint v)\\
  & = \alpha(x,y)\int_G\int_G\bar{\rho}(u,v)f_c(v) P(u,\dint v) K^{k-1}(y,\dint u)\\
 & \qquad +\alpha^c(x,y) \int_G\int_G\bar{\rho}(u,v)f_c(v) P(u,\dint v) 
 K^{k-1}(x,\dint u)\\
 & = 0,
 \end{align*}
 where the last equality follows from
 \[
  \int_G f_c(v) \bar{\rho}(u,v)P(u,\dint v) 
  \overset{\eqref{eq: freq_used_fact}}{=} Z\ \mathbb E_\mu (f_c) = 0
	\qquad
	\forall u \in G.
 \]
 By the same argument we obtain $\mathrm K_\text{aug} h_3 = 0$.
 Hence, for the augmented MH Markov chain $(X_n,Y_n)_{n\in\mathbb{N}}$ 
 condition~\ref{en: maxwell_woodroofe}. of Theorem~\ref{thm: gen_CLT} 
 is satisfied for the function $h_3$  
 and by the inheritance of the Harris recurrence from $K$ to $K_\text{aug}$, 
 see Lemma~\ref{theo:Harris}, 
 we get 
 \[
  \frac{1}{\sqrt{n}} \sum_{k=1}^n h_3(X_k,Y_k) \
 \xrightarrow[n\to\infty]{\mathcal D}\ 
  \mathcal{N}(0,\sigma^2_{S}(h_3)).
 \]
 Here
 \begin{equation*}
   \sigma^2_{S}(h_3) = \Var(h_3(X_1,Y_1)) 
  + 2\sum_{k=1}^\infty \Cov(h_3(X_1,Y_1),h_3(X_{k+1},Y_{k+1})).
 \end{equation*}
By exploiting again the fact that $\mathrm K_\text{aug}^{k} h_3 = 0 $ for 
$k\geq 1$ we obtain
 \begin{equation*}
 \Cov(h_3(X_1,Y_1),h_3(X_{k+1},Y_{k+1})) 
  =  \int_{G\times G} (\mathrm K_\text{aug}^{k} h_3)(x,y) h_3(x,y) \nu(\dint x \dint y)
  = 0,
 \end{equation*}
 such that 
  \begin{equation*}
    \sigma^2_{S}(h_3) = \Var(h_3(X_1,Y_1)) = Z^2 \sigma^2_{A}(f).
 \end{equation*}
 Further,
\[
\sqrt{n}(A_n(f)-\mathbb{E}_\mu(f)) 
= \frac{n^{-1/2} \sum_{j=1}^n h_3(X_j,Y_j)}{\frac{1}{n} \sum_{j=1}^n \bar \rho(X_j,Y_j)}.
\]
The denominator converges by Theorem~\ref{thm: SLLN_gen} to $Z$ 
as well as
\[
  n^{-1/2} \sum_{k=1}^n h_3(X_k,Y_k)\ 
   \xrightarrow[n\to\infty]{\mathcal D}\
 \mathcal{N}(0,Z^2 \sigma^2_{A}(f)), 
\]
such that by Slutsky's Theorem
the assertion is proven.
\end{proof}

\begin{remark}
  It is remarkable that the 
  asymptotic variance $\sigma^2_A(f)$ of $A_n(f)$ coincides 
  with the asymptotic variance of the importance sampling estimator
  \[
   \frac{\sum_{k=1}^n \bar{\rho}(X_k,Y_k) f(Y_k)}{\sum_{k=1}^n \bar \rho(X_k,Y_k)}
  \]
  given \emph{independent}~random variables $(X_k,Y_k)\sim \nu$ for $k\in\mathbb N$, 
  see \cite[Section~2.3.1]{AgPaSaSt17} or \cite[Section~9.2]{Ow13}.
 Here, $\nu$ denotes the stationary measure of the augmented MH Markov
 chain given in \eqref{equ:nu}.
 Hence, the fact that $A_n(f)$ is based on the, 
 in general, dependent sequence $(X_k,Y_k)_{k\in\mathbb N}$ of the augmented MH 
 Markov chain, does surprisingly not effect its asymptotic variance.
\end{remark}

\begin{remark}
 Often it is of interest to estimate the asymptotic variance appearing in a CLT. 
 For a given $f\in L^1(\mu)$ the corresponding 
 quantity, given by Theorem~\ref{thm: CLT}, can be rewritten as 
 \[
  \sigma^2_{A}(f) 
  = 
  \frac{\int_{G\times G} (f(y)-\mathbb{E}_\mu(f))^2 \bar{\rho}(x,y)^2 P(x,\dint y) \mu(\dint x)}
  {\left(\int_{G\times G}  \bar{\rho}(x,y) P(x,\d y) \mu(\d x)\right)^2}.
 \]
 Given this representation of $\sigma^2_{A}(f)$ we suggest estimating it by
 \[
\frac{n \cdot \sum_{k=1}^n \left[ f(Y_k) - \frac{1}{n} \sum_{j=1}^n f(X_j)  \right]^2 
 {\bar \rho}(X_k,Y_k)^2}{
 \left(\sum_{k=1}^n {\bar \rho}(X_k,Y_k)\right)^2}
\]
where $\frac{1}{n} \sum_{j=1}^n f(X_j)$ can also be replaced by $A_n(f)$.
\end{remark}

Now we turn to a non-asymptotic analysis, where the error criterion is the mean squared error.

\subsection{Mean squared error bound}
In this section we provide explicit bounds for the mean squared error of $A_n$.
Those estimates are an immediate consequence of the following two lemmas, which 
are similar to the arguments in \cite[Theorem~2]{MaNo07}
and \cite[Theorem~2.1]{AgPaSaSt17}.

\begin{lemma} \label{lem: decomp_MSE}
 Let $(X_n,Y_n)_{n\in \mathbb{N}}$ denote an augmented MH 
 Markov chain. 
 For $f\colon G\to\mathbb R$ define
 \begin{align*}
  D(f) & := \int_{G\times G} f(y)\bar{\rho}(x,y) P(x,\dint y) \mu(\dint x),\\
  D_n(f) & := \frac{1}{n} \sum_{j=1}^n \bar{\rho}(X_j,Y_j) f(Y_j).
 \end{align*}
 Then, for bounded $f$, i.e., $\norm{f}_\infty := \sup_{x\in G} |f(x)| < \infty$,
we have
 \begin{equation*}
  \mathbb{E}\abs{A_n(f)-\mathbb{E}_\mu(f)}^2
  \leq \frac{2}{D(1)^2}
  \left( \norm{f}_\infty^2 \mathbb{E} \abs{D(1)-D_n(1)}^2 + \mathbb{E}\abs{D_n(f)-D(f)}^2  \right).
 \end{equation*}
\end{lemma}
\begin{proof}
 Observe that $D(1)=Z$.
 Further
 \begin{align*}
   \mathbb{E}\abs{A_n(f)-\mathbb{E}_\mu(f)}^2 & = \mathbb{E}\abs{ \frac{D_n(f)}{D_n(1)}
      - \frac{D(f)}{Z}}^2
  \\
  & 
  = \mathbb{E}\abs{ \frac{D_n(f)}{D_n(1)} 
      - \frac{D_n(f)}{Z}
      + \frac{D_n(f)}{Z}
      - \frac{D(f)}{Z}}^2.    
 \end{align*}
Using the fact that $(a+b)^2 \leq 2 a^2 + 2 b^2 $ for any $a,b\in \mathbb{R}$
gives 
\begin{align*}
   \mathbb{E}\abs{A_n(f)-\mathbb{E}_\mu(f)}^2
& 
\leq 2 \mathbb{E}\abs{ \frac{D_n(f)}{D_n(1)} 
      - \frac{D_n(f)}{Z}}^2
      + 2\mathbb{E}\abs{ \frac{D_n(f)}{Z}
      - \frac{D(f)}{Z}}^2 \\
& = \frac{2}{Z^2} \mathbb{E} \abs{\frac{D_n(f)}{D_n(1)}\left( D_n(1)-Z \right)}^2
+\frac{2 \mathbb{E} \abs{D_n(f)-D(f)}^2}{Z^2}\\
& \leq \frac{2}{Z^2} \left( \norm{f}_\infty^2 \mathbb{E} \abs{D_n(1)-Z }^2 
+ \mathbb{E} \abs{D_n(f)-D(f)}^2
\right).
\end{align*}
\end{proof}

\begin{lemma}
 Assume that the initial distribution 
 is the stationary one, that is, $X_1\sim \mu$. 
 Then, with the notation from Lemma~\ref{lem: decomp_MSE},
 we have
 \begin{equation*}
n\cdot \mathbb{E} \abs{D_n(f)-D(f)}^2
 = \int_{G^2} f(y)^2\bar{\rho}(x,y)^2 P( x, \dint y) \mu(\dint x)
-Z^2 \mathbb{E}_\mu(f)^2.
  \end{equation*}
\end{lemma}
\begin{proof}
 Observe that
 \[
  D(f) = \int_G \int_G f(y) \bar{\rho}(x,y)P(x,\dint y)\mu(\dint x)
\overset{\eqref{eq: freq_used_fact}}{=}
 Z \cdot \mathbb{E}_\mu(f).
 \]
 Define the centered function $g_c(x,y)
 := \bar{\rho}(x,y)f(y) -  Z \cdot \mathbb{E}_\mu(f)$ for any $x,y\in G$.
 We have
 \begin{align*}
\mathbb{E} \abs{D_n(f)-D(f)}^2
& = \frac{1}{n^2} \sum_{j=1}^n \mathbb{E}\left[g_c(X_j,Y_j)^2\right]
+ \frac{2}{n^2} \sum_{j=1}^{n-1} \sum_{i=j+1}^n \mathbb{E}\left[g_c(X_i,Y_i)g_c(X_j,Y_j)\right].
\end{align*}
 Exploiting the fact that the initial distribution is
 the stationary one we obtain
for $i\geq j$ that  
\begin{equation*}
 \mathbb{E}\left[ g_c(X_i,Y_i)g_c(X_j,Y_j) \right]
 = \int_{G\times G} g_c(x,y)\, (K_\text{aug}^{i-j} g_c)(x,y) P(x,\dint y) \mu(\dint x).
\end{equation*}
In the case $k:=i-j> 1$ we have by representation \eqref{eq: for_K_aug_low_est}
that
\begin{align*}
  {\mathrm K}_\text{aug}^k g_c(x,y) 
 & = \alpha(x,y) \int_G \int_G g_c(u,v) P(u,\dint v) K^{k-1}(y,\dint u) \\
 & \qquad \qquad  + \alpha^c(x,y) \int_G \int_G g_c(u,v) P(u,\dint v) K^{k-1}(x,\dint u)
\end{align*}
and
\begin{equation*}
 \int_G g_c(u,v) P(u,\dint v) 
 = \int_G f(v) \bar{\rho}(u,v) P(u,\dint v) -  Z \cdot \mathbb{E}_\mu(f)
 \overset{\eqref{eq: freq_used_fact}}{=} 0
\end{equation*}
leads to ${\mathrm K}^k_{\text{aug}} g_c(x,y)= 0$. 
By similar arguments we obtain ${\mathrm K}_\text{aug} g_c(x,y) = 0$. 
Hence,
\begin{align*}
 \mathbb{E}\abs{D_n(f)-D(f)}^2
& = \frac{1}{n} \mathbb{E}\left[ g_c(X_1,Y_1)^2\right]\\ 
& =  \frac{1}{n}  \left(\int_{G^2} f(y)^2\bar{\rho}(x,y)^2 P( x, \dint y) \mu(\dint x)
  -Z^2 \mathbb{E}_\mu(f)^2 \right).
\end{align*}
\end{proof}

By the combination of both lemmas we derive the following theorem.
\begin{theorem}  \label{thm: mse}
 Assume that the initial distribution of an augmented MH 
 Markov chain $(X_n,Y_n)_{n\in\mathbb{N}}$
 is the stationary one, i.e., $X_1\sim \mu$. 
 Then, for bounded $f\colon G\to \mathbb R$ we obtain
 \begin{equation*}
    \mathbb{E}\abs{A_n(f)-\mathbb{E}_\mu(f)}^2
  \leq \frac{4}{n} \norm{f}_\infty^2 
  \int_{G\times G} \frac{\dint \mu}{\dint P(x,\cdot)}(y) \mu(\dint y) \mu(\dint x).
   \end{equation*}
\end{theorem}
\begin{remark}
 Let us mention here two things: First, we assumed that the initial distribution is
 the stationary one. This assumption is certainly restrictive, we refer to 
 \cite{LaMiNi13,Ru12} for techniques to derive explicit error bounds for more general initial
 distribution. Second, the factor 
 \[4 \norm{f}_\infty^2 
  \int_{G\times G} \frac{\dint \mu}{\dint P(x,\cdot)}(y) \mu(\dint y) \mu(\dint x)\] 
  in the estimate is an upper bound of the asymptotic variance $\sigma^2_{A}(f)$ derived in Theorem~\ref{thm: CLT}. We conjecture that the estimate actually 
  holds with $\sigma^2_{A}(f)$ instead of this upper bound.
\end{remark}

\subsection{Optimal calibration of proposals}\label{sec:stepsize}
Given the explicit expression for the asymptotic variance $\sigma^2_A(f)$ involving the proposal kernel $P$, we can ask for an optimal choice of the kernel $P\colon G \times \mathcal G \to [0,1]$ in order to minimize $\sigma^2_A(f)$.
However, finding an optimal kernel among all admissible kernels is, in general, an infeasible task.
In practice, one often considers common types of proposal kernels $P = P_s$ with a tunable stepsize parameter $s > 0$ and ask for the optimal value of $s$.
For example, given a measure $\mu$ on $G \subseteq \mathbb{R}^d$, we can use the \emph{random walk} proposal
\begin{equation}\label{equ:RW_prop}
	P_s(x, \cdot)
	= \mathcal{N}(x, s^2 C),
	\qquad
	s > 0,
\end{equation}
where $x\in \mathbb R^d$ and $C\in \mathbb{R}^{d\times d}$ denotes a covariance matrix, within a MH algorithm.
For this proposal and the classical path average estimator $S_n(f)$ it is widely known that a good stepsize $s^*_S$ is chosen in such a way that the average acceptance rate is
\[
\int_G \alpha(x,y) P_{s^*_S}(x,\d y) \mu(\d x)
\approx
0.234.
\]
For a justification and further details we refer to \cite{RoRo01}.
For the MH importance sampling estimator $A_n(f)$ we look for an optimal stepsize parameter $s^*_A$. Optimal in the sense that it 
minimizes the asymptotic variance of $A_n(f)$, thus, we ask for
\begin{align*}
s^*_A 
& :=
\argmin_{s > 0} 
V(s),
\quad
V(s) := 
\int_G \int_G (f(y)-\mathbb{E}_\mu(f))^2 \frac{\rho(y)}{p_s(x,y)} \, \mu(\dint y) \mu(\dint x),
\end{align*}
where $p_s(x,\cdot)$ denotes the density of $P_s(x,\cdot)$ w.r.t.~the reference measure $\mu_0$.
If we assume that the mapping $s\mapsto p_s(x,y)$ is differentiable for each $(x,y)\in G\times G$ with derivative $\frac{\d }{\d s} p_s(x,y)$, then any $s$ minimizing $V(s)$ satisfies
\begin{align}\label{equ:Optimal_s_cond}
	0
	& 
	= 
	\frac{\d }{\d s} V(s)
	=
	\int_G \int_G (f(y)-\mathbb{E}_\mu(f))^2 \rho(y) \frac{\frac{\d }{\d s} p_s(x,y)}{p^2_s(x,y)} \, \mu(\dint y) \mu(\dint x).
\end{align}
By the fact that $\mu(\dint y) \mu(\dint x) \propto \bar{\rho}_s(x,y)\ P_s(x, \dint y) \mu(\dint x)$, where $\bar{\rho}_s(x,y) = \frac{\rho(y)}{p_s(x,y)}$, we can rewrite \eqref{equ:Optimal_s_cond} and approximate $\frac{\d }{\d s} V(s)$ by using $(X_k,Y_k)$, $k=1,\ldots,n$ from the augmented Markov chain. Thus
\begin{align*}\label{equ:Optimal_s_cond_emp}
	0 & = \int_G \int_G (f(y)-\mathbb{E}_\mu(f))^2\ \bar{\rho}_s^2(x,y)\ \frac{\frac{\d }{\d s} p_s(x,y)}{p_s(x,y)} \, P_s(x, \dint y) \mu(\dint x)\\	
	& \approx
	\frac 1n \sum_{k=1}^n
	\left(f(Y_k) - \frac 1n \sum_{j=1}^n f(X_j)\right)^2 
	\bar{\rho}_s^2(X_k,Y_k) \frac{\frac{\d }{\d s} p_s(X_k,Y_k)}{p_s(X_k,Y_k)}.
\end{align*}
In practice we can calibrate $s$ such that the empirical average on the right-hand side is close to zero.
We demonstrate the feasibility of this approach for two common proposals.

\begin{example}[Optimal calibration of the random walk-MH]
We consider $\mu_0$ as the Lebesgue measure on $G \subseteq\mathbb{R}^d$ and $P_s$ as in \eqref{equ:RW_prop}.
Thus,
\[
	p_s(x,y) = \frac1{s^d \sqrt{\det(2\pi C)} } \exp\left(-\frac {|y-x|^2_{C}}{2s^2}\right)
\]
where $|y-x|^2_{C} := (y-x)^\top C^{-1}(y-x) $, and
\begin{align*}
	\frac{\d }{\d s} p_s(x,y) 
	& = \left( -d s^{-d-1} + s^{-d-3} |y-x|^2_{C}\right) \frac{\exp\left(-\frac {|y-x|^2_{C}}{2s^2}\right)}{\sqrt{\det(2\pi C)} }\\
	& =  \left( -d s^{-1} + s^{-3} |y-x|^2_{C}\right) p_s(x,y).
\end{align*}
Hence, the necessary condition \eqref{equ:Optimal_s_cond} boils down to
\begin{align*}
	0
	=
	\int_G \int_G (f(y)-\mathbb{E}_\mu(f))^2\ \bar{\rho}_s^2(x,y)\ \left( -d s^{-1} + s^{-3} |y-x|^2_{C}\right)  \, P_s(x, \dint y) \mu(\dint x),
\end{align*}
which can be rewritten as
\begin{align*}
	s^2
	=
	\frac{\int_G \int_G (f(y)-\mathbb{E}_\mu(f))^2\ \bar{\rho}_s^2(x,y)\ |y-x|^2_{C}  \, P_s(x, \dint y) \mu(\dint x)}{d \int_G \int_G (f(y)-\mathbb{E}_\mu(f))^2\ \bar{\rho}_s^2(x,y)  \, P_s(x, \dint y) \mu(\dint x)}.
\end{align*}
In practice, we then can seek an $s>0$ such that for $n$ states $(X_k,Y_k)$ of the augmented MH Markov chain generated by the proposal $P_s(x,\cdot) = \mathcal{N}(x, s^2 C)$ we have
\begin{equation}\label{equ:Optimal_s_RW_emp}
	s^2
	\approx
	\frac{\sum_{k=1}^n
	\left(f(Y_k) - \frac 1n \sum_{j=1}^n f(X_j)\right)^2 
	\bar{\rho}_s^2(X_k,Y_k) \ |Y_k-X_k|^2_{C} }{d \sum_{k=1}^n
	\left(f(Y_k) - \frac 1n \sum_{j=1}^n f(X_j)\right)^2 
	\bar{\rho}_s^2(X_k,Y_k)}.
\end{equation}
\end{example}

\begin{example}[Optimal calibration of the MALA]
Another common proposal on $G=\mathbb{R}^d$ is the one of the \emph{Metropolis-adjusted Langevin algorithm (MALA)}, given by 
\begin{equation}\label{equ:MALA_prop}
	P_s(x,\cdot) = \mathcal{N}\left( x + \frac{s^2}{2}\nabla \log \rho(x), s^2 I_d \right),
\end{equation}
where we assume that $\log \rho\colon G \to \mathbb{R}$ is differentiable and $I_d$ denotes the identity matrix in $\mathbb{R}^d$.
The resulting proposal density is
\[
	p_s(x,y) = \frac1{s^{d} (2\pi)^{d/2}} \exp\left(-\frac {|y - m_s(x)|^2}{2s^2}\right),
\]
with $m_s(x) := x + \frac{s^2}{2}\nabla \log \rho(x)$.
In order to compute the derivative $\frac{\d }{\d s} p_s(x,y) $ we require
\[
	\frac{\d }{\d s} |y - m_s(x)|^2 
	=
	- 2s (y - m_s(x))^\top \nabla \log \rho(x),
\]
which then yields
\begin{align*}
	\frac{\d }{\d s} p_s(x,y) 
	& = \left( -d s^{-1} + s^{-3} |y-m_s(x)|^2 + s^{-1} (y - m_s(x))^\top \nabla \log \rho(x)  \right)\, p_s(x,y).
\end{align*}
Thus, in the case of MALA the necessary condition \eqref{equ:Optimal_s_cond} is equivalent to
\begin{align*}
	s^2
	=
	\frac{\int_G \int_G (f(y)-\mathbb{E}_\mu(f))^2\ \bar{\rho}_s^2(x,y)\ |y-m_s(x)|^2  \, P_s(x, \dint y) \mu(\dint x)}{\int_G \int_G (f(y)-\mathbb{E}_\mu(f))^2\ \bar{\rho}_s^2(x,y)\ [d - (y - m_s(x))^\top \nabla \log \rho(x)]  \, P_s(x, \dint y) \mu(\dint x)}.
\end{align*}
Again, in practice we seek for an $s>0$ such that given $n$ states $(X_k,Y_k)$ of the augmented MH Markov chain generated by the MALA proposal $P_s$ in \eqref{equ:MALA_prop} we have $s^2$ close to
\begin{equation}\label{equ:Optimal_s_MALA_emp}
	\frac{\sum_{k=1}^n
	\left(f(Y_k) - \frac 1n \sum_{j=1}^n f(X_j)\right)^2 
	\bar{\rho}_s^2(X_k,Y_k) \ |Y_k-m_s(X_k)|^2 }{\sum_{k=1}^n
	\left(f(Y_k) - \frac 1n \sum_{j=1}^n f(X_j)\right)^2 
	\bar{\rho}_s^2(X_k,Y_k) \ [d - (Y_k - m_s(X_k))^\top \nabla \log \rho(X_k)]}.
\end{equation}
\end{example}


\section{Numerical examples}\label{sec_num_exam}
We want to illustrate the benefits as well as the limitations of the MH importance sampling estimator $A_n(f)$ at two simple but representative examples.
To this end, we compare the considered $A_n(f)$ to the classical path average estimator $S_n(f)$ as well as to two other established estimators using also the proposed states $Y_k$ generated in the MH algorithm. Namely
\begin{itemize}
\item
the \emph{waste-recycling Monte Carlo estimator}, for further details we refer to \cite{Fr04,Fr06,DeJo09}, given by
\[
	WR_n(f)
	:=	
	\sum_{k=1}^n
	\left( 1 - \alpha(X_k,Y_k) \right)  f(X_k)
	+
	\alpha(X_k,Y_k) f(Y_k);
\]
\item
another Markov chain importance sampling estimator also based on the proposed states, see \cite{SchKle18}, given by
\[
	B_n(f)
	:=	
	\frac{\sum_{k=1}^n \widetilde w_n(X_{1:n}, Y_k)  f(Y_k)}{\sum_{k=1}^n \widetilde w_n(X_{1:n}, Y_k)},
	\qquad
	\widetilde w_n(X_{1:n}, Y_k)
	:=
	\frac{\rho(Y_k)}{\sum_{j=1}^n p(X_j, Y_k)}.
\]
\end{itemize}
The notation $X_{1:n}$ within $B_n(f)$ stands for $X_1,\dots,X_n$. In the following we provide two comments w.r.t. $B_n$ and the other estimators.

\begin{remark}
For the convenience of the reader
we justify heuristically that $B_n(f)$ approximates $\mathbb{E}_\mu(f)$.
For this let $\nu_Y$ be the marginal distribution of the stationary probability measure $\nu$ on $G\times G$ of the augmented Markov chain $(X_k,Y_k)_{k\in\mathbb{N}}$, that is, 
\[
\nu_Y(\d y) := \int_G \nu(\d x \d y) = \int_G P(x,\d y) \mu(\d x).
\]
Intuitively $\nu_Y$ can be considered as the asymptotic distribution of the proposed states. The empirically computed weights $\widetilde w_n(X_{1:n}, Y_k)$ in $B_n(f)$ approximate importance sampling weights $\widetilde w(Y_k) \propto \frac{\d \mu}{\d \nu_Y}(Y_k)$ resulting from the asymptotic distribution $\nu_Y$ of the proposed states $Y_k$. Now if we substitute $\widetilde w_n(X_{1:n}, Y_k)$ within $B_n(f)$ by $\widetilde w(Y_k)$ we have an importance sampling estimator based on the proposed states which approximates $\mathbb{E}_\mu(f)$.
\end{remark}

\begin{remark}
By the fact that within the estimators $S_n(f)$, $A_n(f)$, and $WR_n(f)$ only one sum appears, the number of arithmetic operations and therefore the complexity is $\mc O(n)$. In contrast to that, within the alternative Markov chain importance sampling estimator $B_n(f)$ an additional sum appears in the computation of each weight $w_n(X_{1:n}, Y_k)$, such that the overall number of arithmetic operations is $\mc O(n^2)$, since 
we need to compute $w_n(X_{1:n}, Y_k)$, $k=1,\ldots,n$.
To take this into account, we often compare the former three estimators to $B_{\sqrt{n}}(f)$.
Besides that an optimal tuning of the proposal stepsize of the estimator $B_n(f)$ is left open in \cite{SchKle18}, however, the authors suggest to simply use the usual calibration rule for the classical path average estimator $S_n(f)$ from \cite{RoRo01}.
\end{remark}

\subsection{Bayesian inference for a differential equation} \label{subsec_num_exam_1}
We consider a boundary value problem in one spatial dimension $x\in[0,1]$ which serves as a simple model for, e.g., stationary groundwater flow:
\begin{equation} \label{equ:1D_PDE}
	-\frac {\d}{\d x} \left( \exp(u_1) \,\frac {\d}{\d x}p(x) \right) = 1, 
	\quad p(0) = 0, \; p(1) = u_2.
\end{equation}
Here, the unknown parameters $u = (u_1,u_2)$ involving the log-diffusion coefficient $u_1$ and the Dirichlet data $u_2$ at the righthand boundary $x=1$ shall be inferred given noisy observations $y\in\mathbb R^2$ of the solution $p$ at $x_1 = 0.25$ and $x_2 = 0.75$.
This inference setting has been already applied as a test case for sampling and filtering methods in \cite{ErnstEtAl2015, GarbunoInigoEtAl2019, HertyVisconti2019}.
We place a Gaussian prior on $u=(u_1,u_2)$, namely, $\mu_0 \sim N(0, I_2)$ where $I_2$ denotes the identity matrix in $\mathbb R^2$.
The observation vector is given by $y=(27.5, 79.7)$ and we 
assume an additive measurement 
noise $\varepsilon \sim N(0,0.01 I_2)$, i.e., the likelihood $L(y|u)$ of observing $y$ given a fixed value $u\in \mathbb R^2$ is
\[
	L(y | u) :=  \frac{100}{2\pi} \exp\left(- \frac{100}2  \|y - F(u)\|^2 \right)
\]
where $\Vert\cdot\Vert$ denotes the Euclidean norm and $F\colon \mathbb R^2\to \mathbb R^2$ the mapping $(u_1,u_2) \mapsto (p(x_1), p(x_2))$ with
\[
 p(x)=u_2x+\frac{\exp(-u_1)}{2}(x-x^2), \qquad x\in [0,1].
\]
The resulting posterior measure for $u$ given the observation $y$ follows then the form \eqref{equ:mu} with $\rho(u) := L(y | u)$.
The negative log prior and posterior density are presented in Figure~\ref{fig:exam1}.

\begin{figure*}[h]
	\includegraphics[width=0.48\textwidth]{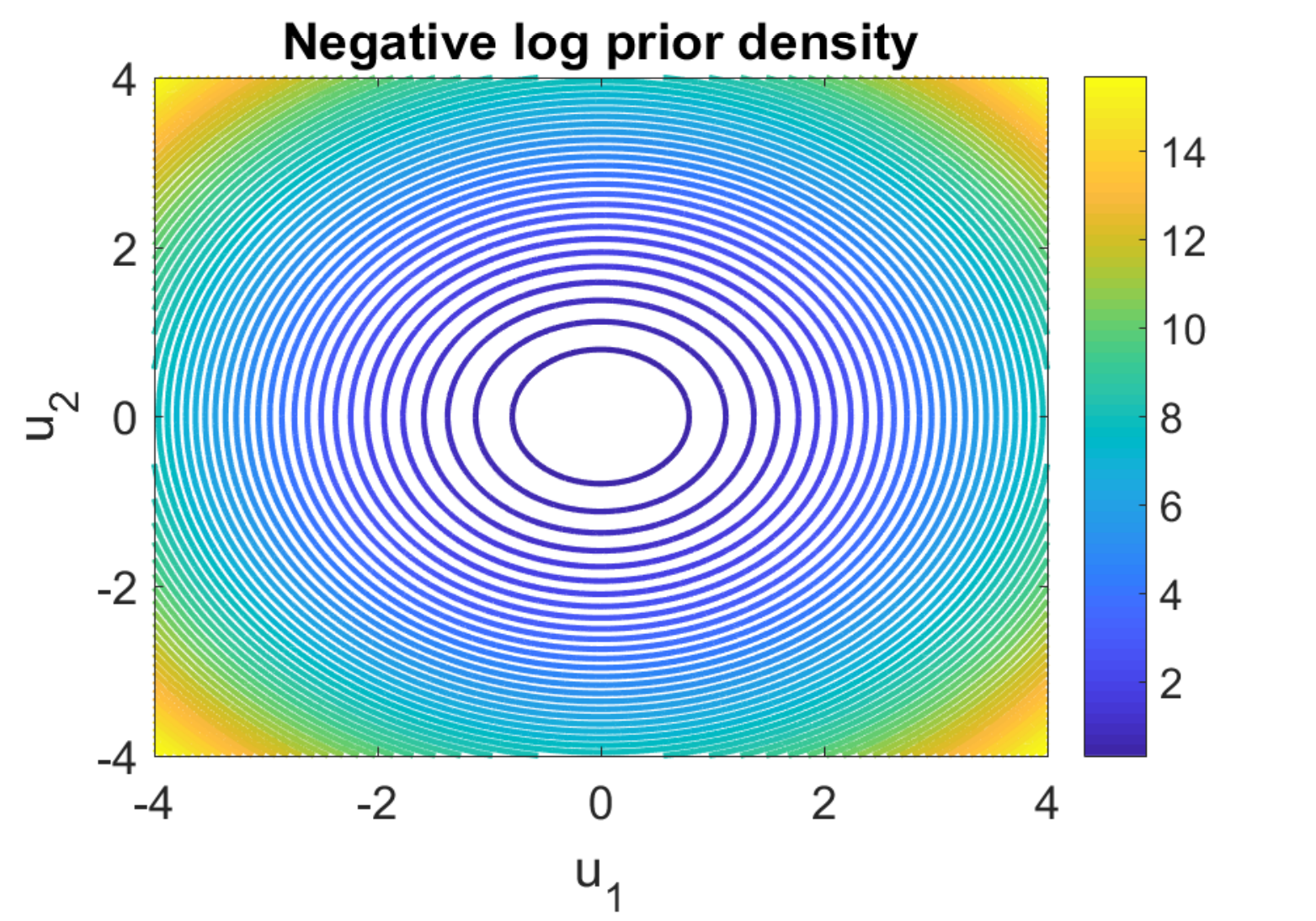}
	\includegraphics[width = .48\textwidth]{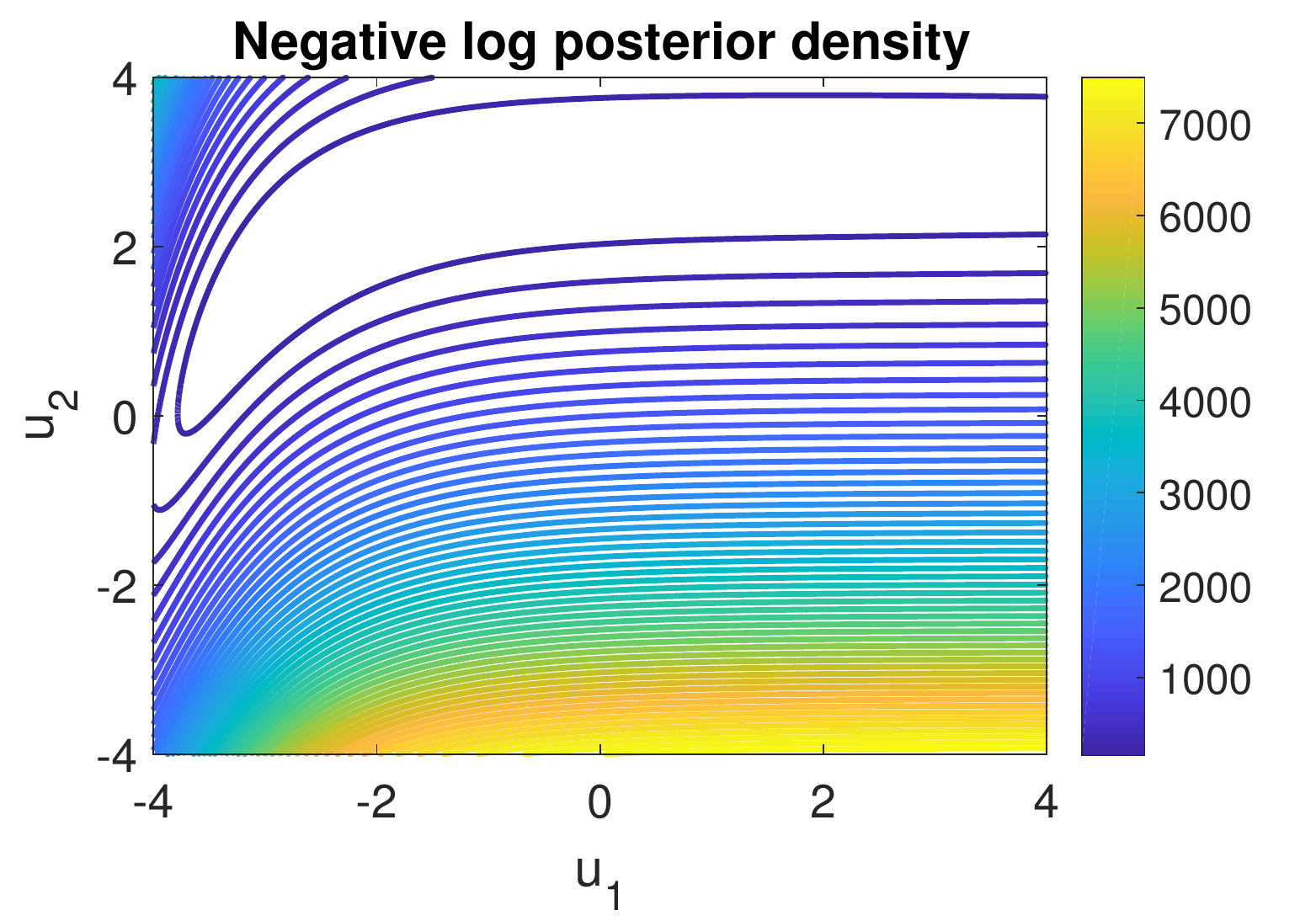}
\caption{Contour plot of the normal prior and the resulting posterior density 
for the example of Section~\ref{subsec_num_exam_1}.}
\label{fig:exam1}
\end{figure*}

For approximate sampling of the posterior $\mu$ we apply now the RW-MH algorithm and MALA, see Section \ref{sec:stepsize}, with various values of the stepsize $s$ .
We let the Markov chains run for $n=10^4$ iterations after a burn-in of $n_0=10^3$ iterations.
Then, we use the generated path of the (augmented) MH Markov chain in order to estimate the posterior mean $\mathbb E_\mu(\mathbf f)$ where $\mathbf f(u) :=(f_1(u),f_2(u))$ with $f_i(u):=u_i$.
We approximate $\mathbb E_\mu(\mathbf f)$ by the various estimators $E_n(\mathbf f)$ discussed at the begining of this section, that is,
\[
	E_n(\mathbf f) \in \{S_n(\mathbf f), A_n(\mathbf f), B_n(\mathbf f), B_{\sqrt{n}}(\mathbf f), WR_n(\mathbf f)\}.
\]
The true value $\mathbb E_\mu(\mathbf f)$ of the posterior mean is computed by Gauss quadrature employing $1500$ Gauss--Hermite nodes in each dimension which ensures a quadrature error smaller than $10^{-4}$.
For each choice of the step size $s$ we run $M=1,200$ independent Markov chains and, thus, compute $M$ realizations 
of the estimators $S_n(\mathbf f)$, $A_n(\mathbf f)$, $B_n(\mathbf f)$, $B_{\sqrt{n}}(\mathbf f)$, and $WR_n(\mathbf f)$, respectively.
We use these $M$ realizations in order to empirically estimate the root mean squared error (RMSE) 
\begin{align*}
\mathrm{RMSE}_{E_n}(\mathbf f) 
& := \left(\mathbb{E}\left\|E_n(\mathbf f) - \mathbb{E}_\mu(\mathbf f)\right\|^2 \right)^{1/2},
\end{align*}
of the various estimators $E_n(\mathbf f) \in \{S_n(\mathbf f), A_n(\mathbf f), B_n(\mathbf f), B_{\sqrt{n}}(\mathbf f), WR_n(\mathbf f)\}$
for each chosen stepsize $s$ of the proposal kernels.
The results are displayed in Figure~\ref{fig:exam_results1} and Figure~\ref{fig:exam_results2}, respectively.\footnote{We note that in each setting the squared norm of the bias of the estimators $E_n(f_i)$ is roughly the same size as their variance, i.e., the magnitude of the displayed RMSE coincides basically with $\sqrt 2$ times the standard deviation of the corresponding estimator.
For a larger sample size $n$ the percentage of the bias in the RMSE would have decreased, however, the computation of $B_n(\mathbf f)$ would have become unfeasible.}

\begin{figure*}[h]
	\includegraphics[width = .48\textwidth]{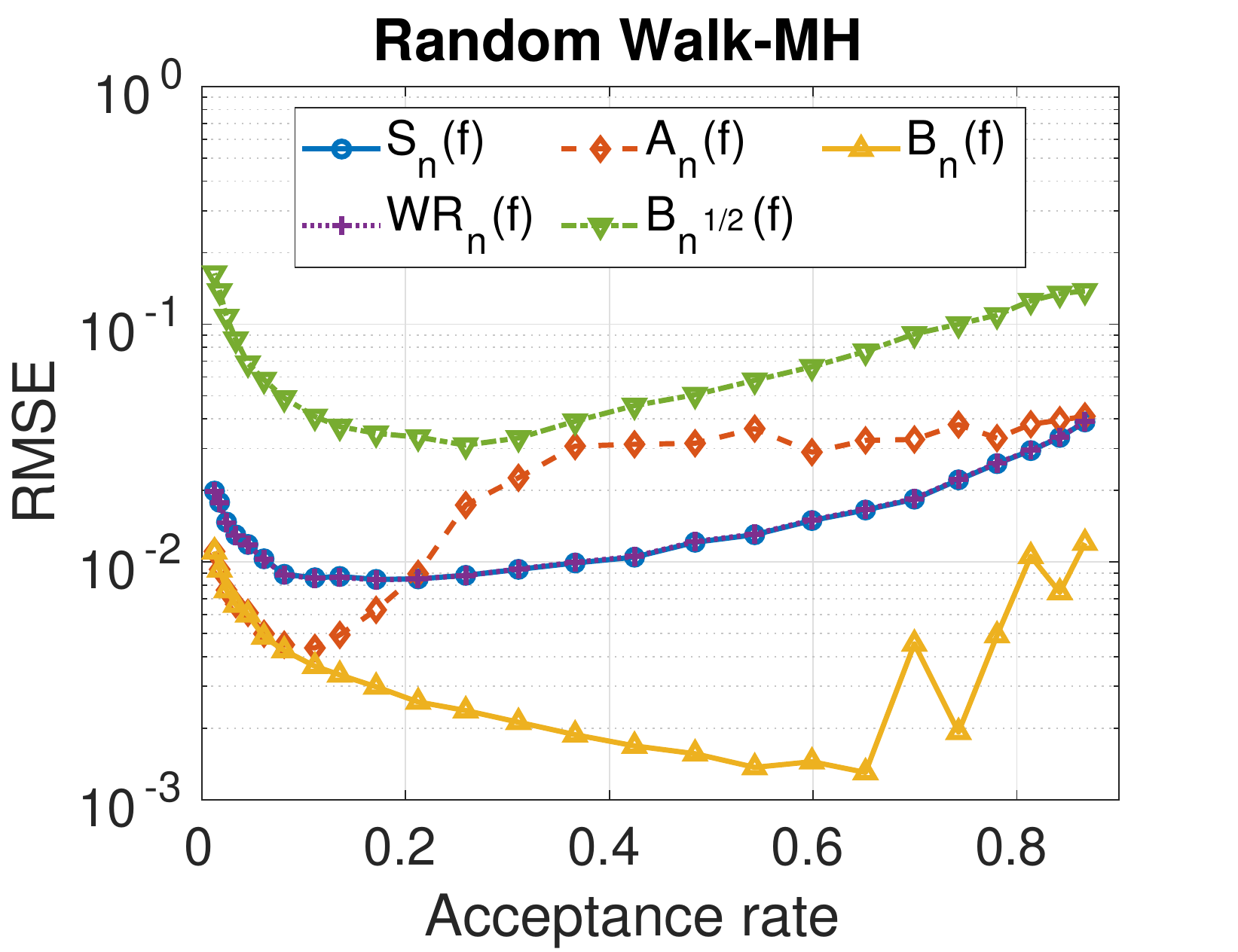}
	\includegraphics[width = .48\textwidth]{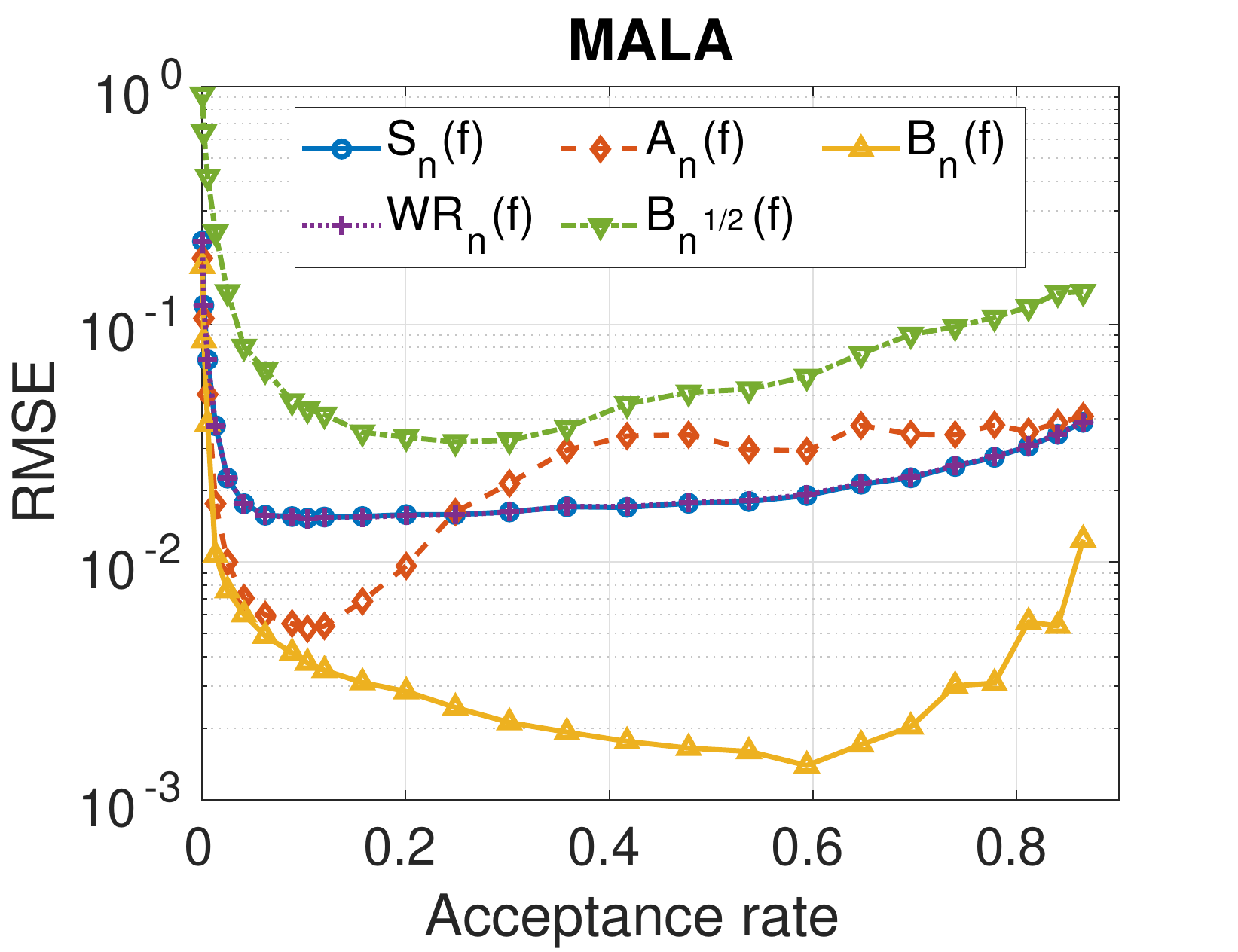}
\caption{RMSE for computing the posterior mean w.r.t.~average acceptance rate for the example {of Section~\ref{subsec_num_exam_1}}.}
\label{fig:exam_results1}
\end{figure*}

\emph{Comparison of $A_n(\mathbf f)$ to $S_n(\mathbf f)$:}
In Figure~\ref{fig:exam_results1} we observe that for a certain range of $s$, the MH importance sampling estimator $A_n(\mathbf f)$
provides a significant error reduction for both proposal kernels, the random-walk and MALA.
In particular, the global minimum for the error of $A_n(\mathbf f)$ is smaller than for $S_n(\mathbf f)$.
In fact, it is roughly half the size for both proposals.
Hence, given the optimal step size $s$ the MH importance sampling method can indeed outperform the classical path average estimator.
In this example, we could reduce the RMSE by $50\%$ without a significant additional cost.
We comment below on how to find this optimal stepsize for $A_n(\mathbf f)$.\par
\emph{Comparison of $A_n(\mathbf f)$ to other estimators:}
We observe in Figure~\ref{fig:exam_results1} that the waste recycling estimator $WR_n(\mathbf f)$ basically coincides with the classical path average estimator $S_n(\mathbf f)$ for both proposals and all chosen stepsizes $s$, i.e., it yields no improvement and is outperformed by $A_n(\mathbf f)$.
Concerning the Markov chain importance sampling estimators $B_n(\mathbf f)$ we obtain a further improvement on $A_n(\mathbf f)$ and nearly can reduce the RMSE by an order of magnitude compared to $S_n(\mathbf f)$.
However, this performance comes at the price of a signicant larger complexity.
If we consider the Markov chain importance sampling estimators $B_{\sqrt n}(\mathbf f)$ with the same complexity as the other estimators $S_n(\mathbf f)$, $A_n(\mathbf f)$, and $WR_n(\mathbf f)$, we in fact observe a worse performance to the other estimators for all chosen stepsizes.
Thus, in the error-vs-complexity sense the estimator $A_n(\mathbf f)$ performs best among all considered estimators if calibrated correctly.\par
\emph{Optimal calibration of $A_n(\mathbf f)$:}
Concerning the optimal stepsize for $A_n(\mathbf f)$ we present in Figure~\ref{fig:exam_results2} a verification of the approach outlined in Section \ref{sec:stepsize}.
For both MH algorithms, the random walk-MH and MALA, we display in the top row the RMSE of $S_n(\mathbf f)$ and $A_n(\mathbf f)$ w.r.t.~the chosen stepsizes.
In the bottom row we display for each stepsize value $s$ the relation of $s^2$ to 
the empirical functionals
\[
	J_{\mathbf f}(s)
	:=
	\frac{\sum_{k=1}^n
	\left|\mathbf f(Y_k) - \frac 1n \sum_{j=1}^n \mathbf f(X_j)\right|^2 
	\bar{\rho}_s^2(X_k,Y_k) \ |Y_k-X_k|^2_{C} }{d \sum_{k=1}^n
	\left|\mathbf f(Y_k) - \frac 1n \sum_{j=1}^n \mathbf f(X_j)\right|^2 
	\bar{\rho}_s^2(X_k,Y_k)}	
\]
and 
\[
	J(s)
	:=
	\frac{\sum_{k=1}^n
	\bar{\rho}_s^2(X_k,Y_k) \ |Y_k-X_k|^2_{C} }{d \sum_{k=1}^n
	\bar{\rho}_s^2(X_k,Y_k)}	
\]
for the random walk-MH and the corresponding $J_{\mathbf f}(s)$ and $J(s)$ for MALA based on \eqref{equ:Optimal_s_MALA_emp}.
In Section \ref{sec:stepsize} we derived as a necessary condition for the optimal stepsize $s_\star$ that $s_\star^2 \approx J_{\mathbf f}(s_\star)$ for both kind of proposals.
Here, we can indeed verify this condition: the optimal $s_\star$, which was calibrated by hand following the rule $s_\star^2 \approx J_{\mathbf f}(s_\star)$, shows indeed also the smallest RMSE in the top row.
The optimal $s_\star$, $J_{\mathbf f}(s_\star)$, and its RMSE are highlighted by a green marker in Figure~\ref{fig:exam_results2}.
Besides that, choosing the rather ``objective'' functional $J(s)$, which is independent of the particular quantity of interest $\mathbf f$, and apply the alternative calibration rule $s_\star^2 \approx J(s_\star)$ does not yield to a stepsize with minimal RMSE for $A_n(\mathbf f )$ --- although the alternatively calibrated stepsize and the resulting RMSE are not that far off from the true optimum.
In summary, Figure~\ref{fig:exam_results2} verifies that the approach in Section \ref{sec:stepsize} can indeed be applied in practice for finding the optimal stepsize for the MH importance sampling estimator $A_n(\mathbf f)$.

\begin{figure*}[h]
	\includegraphics[width = .49\textwidth]{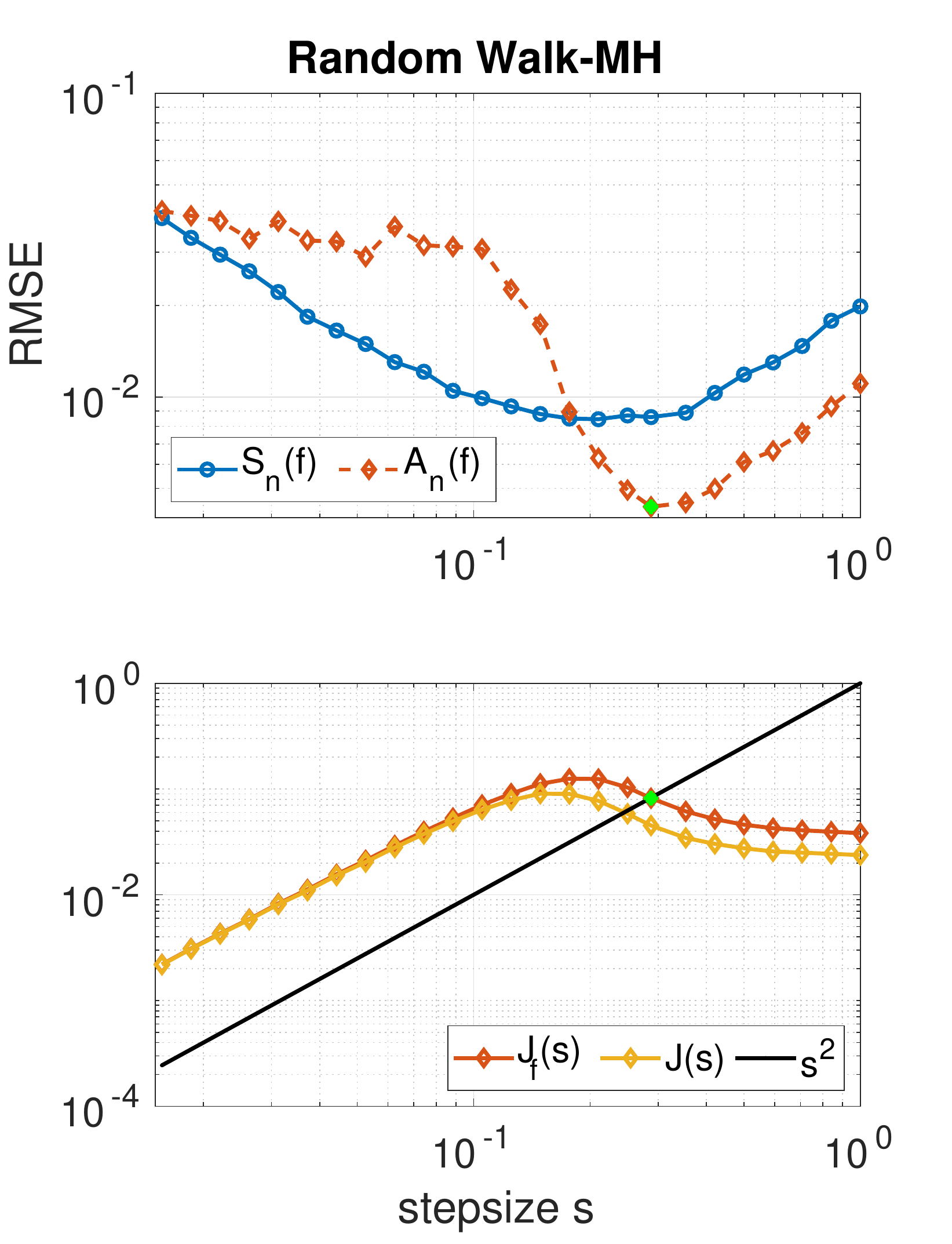}
	\includegraphics[width = .49\textwidth]{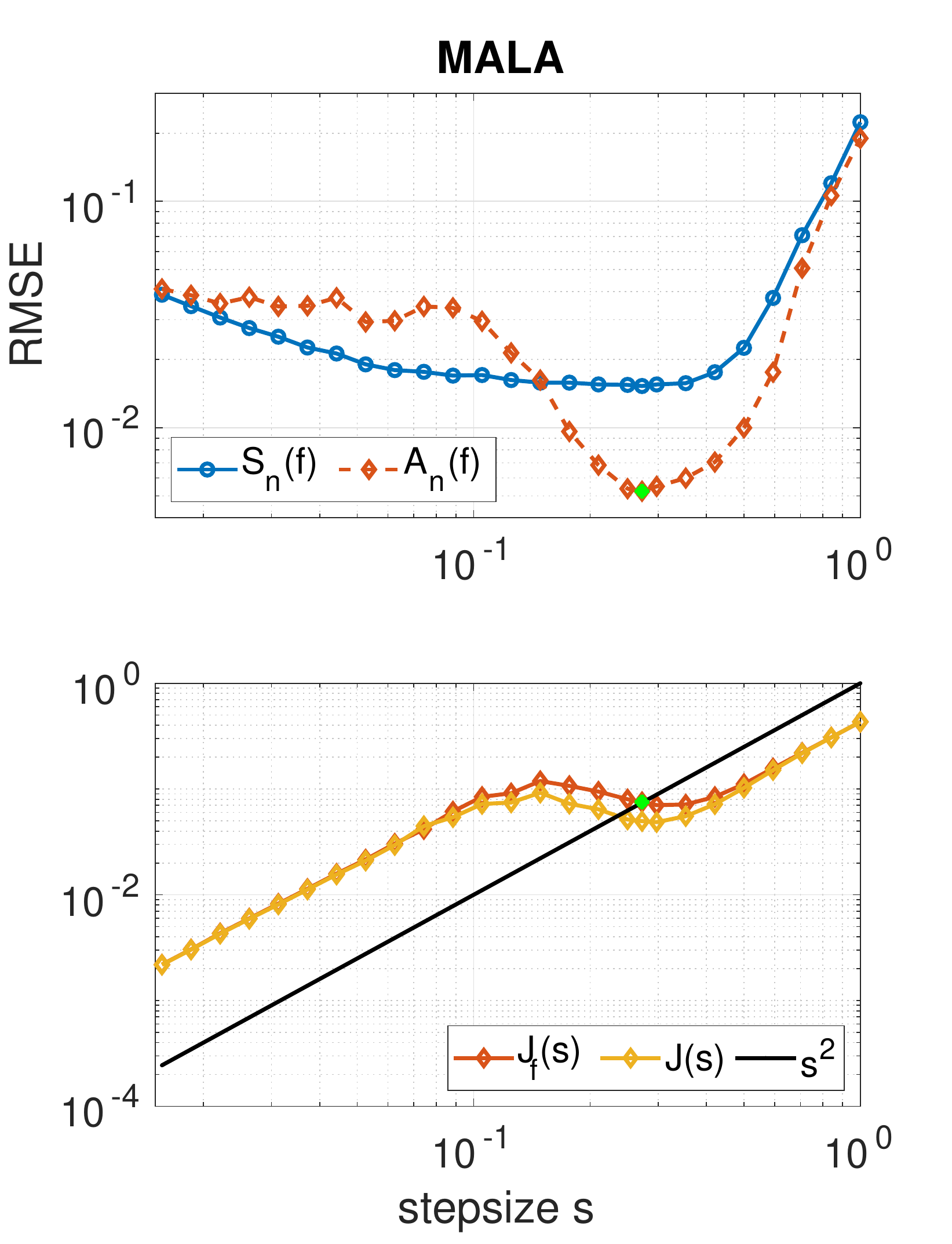}
\caption{RMSE for mean w.r.t.~step size $s$ for the example {of Section~\ref{subsec_num_exam_1}}.}
\label{fig:exam_results2}
\end{figure*}

\subsection{Bayesian inference for probit regression (PIMA data)} \label{subsec_num_exam_2}
The second example is a test problem for logistic regression, see, e.g., \cite{ChopinRidgway15} for a discussion.
Here, nine predictors $x_i \in \mathbb R^9$ such as diastolic blood pressure, body mass index, or age are fitted to the binary outcome $y_i\in\{-1,1\}$ for diagnosing diabetes for $N=768$ members $i =1,\ldots,N$, of the Pima Indian tribe.
For more details about the data we refer to \cite{SmithEtAl88}.
Following \cite{ChopinRidgway15} the likelihood $L(y|\beta)$ for the 
outcome $y \in \{-1,1\}^N$ of the diagnosis is modeled by
\[
	L(y | \beta) := \prod_{i=1}^N \Phi(y_i \beta^\top x_i),
\]
where $\Phi$ denotes the cumulative distribution function of a univariate standard normal distribution and $\beta \in \mathbb R^9$ the unknown regression coefficients (including the intercept).
Moreover, we take independent Gaussian priors for each component of $\beta$ as suggested in \cite{ChopinRidgway15}, i.e., the prior is $\mu_0 = \mathcal N(0, \Lambda)$ where $\Lambda = \diag(\lambda_1, \ldots, \lambda_9)$ with $\lambda_1 = 20$ and $\lambda_i = 5$ for $i\geq 2$.
Given the data set $(x_i,y_i)_{i=1}^N \in \mathbb R^{10\times N}$ the resulting posterior for $\beta$ 
is of the form \eqref{equ:mu} with $\mu_0 = \mathcal N(0, \Lambda)$ and
\[
	\rho(\beta) := \prod_{i=1}^N \Phi(y_i \beta^\top x_i).
\]

For this example we test the performance of the MH importance sampling estimator 
in several dimensions $d = 2,\ldots,9$.
To this end, we modify the regression model for each $d$ by 
setting $\beta = (\beta_1,\ldots,\beta_d,0,\ldots,0) \in \mathbb R^9$ 
and only infer the values of the components $\beta_i$ for $i=1,\ldots,d$.
Hence, the posterior from which we would like to sample is a measure on $\mathbb R^d$, $d=2,\ldots,9$. 
For each $d = 2,\ldots,9$ we perform the same simulations as in the first example, i.e., we generate Markov chains by the MH algorithm using the Gaussian random walk proposal from Section~\ref{subsec_num_exam_1} with varying step size parameter $s$.
Then, we compute the estimates $E_n(\mathbf f):=(E_n(f_i))_{i=1,\dots,d}$ for $\mathbf f(\beta) = (f_i(\beta))_{i=1,\dots,d}$ with $f_i(\beta)=\beta_i$ where $E_n(\mathbf f)$ is again a placeholder for the particular estimator at choice, i.e., $E_n(\mathbf f) \in \{S_n(\mathbf f), A_n(\mathbf f), B_n(\mathbf f), WR_n(\mathbf f), B_{\sqrt{n}}(\mathbf f)\}$ as in Section~\ref{subsec_num_exam_1}.
For each choice of the step size $s$ we repeat this procedure $M=1,200$ times and use the results to compute empirical estimates for the total variance of the estimators $E_n(\mathbf f)$
given by
\begin{align*}
	\Var({E}_n(\mathbf f)) & := \sum_{i=1}^d \Var(E_n(f_i)).
\end{align*}
The number of iterations of the MH Markov chain as well as the burn-in length 
are the same as in Section \ref{subsec_num_exam_1}.
In Figure \ref{fig:exam_results_1_probit} we present for several choices of $d$ the resulting plots 
for the total variance of the estimators w.r.t.~average acceptance 
rate in the MH algorithm, similar to Figure \ref{fig:exam_results1} in the previous section.
Furthermore, Table~\ref{tab:rates} displays the ratios of the total variances $\Var({E}_n(\mathbf f))/\Var({S}_n(\mathbf f))$ for various estimators $E_n(\mathbf f))$ in dimensions $d=2,\ldots,9$.
Our results can be summarized as follows.\par

\begin{figure*}[h]
\centering \includegraphics[width = .48\textwidth]{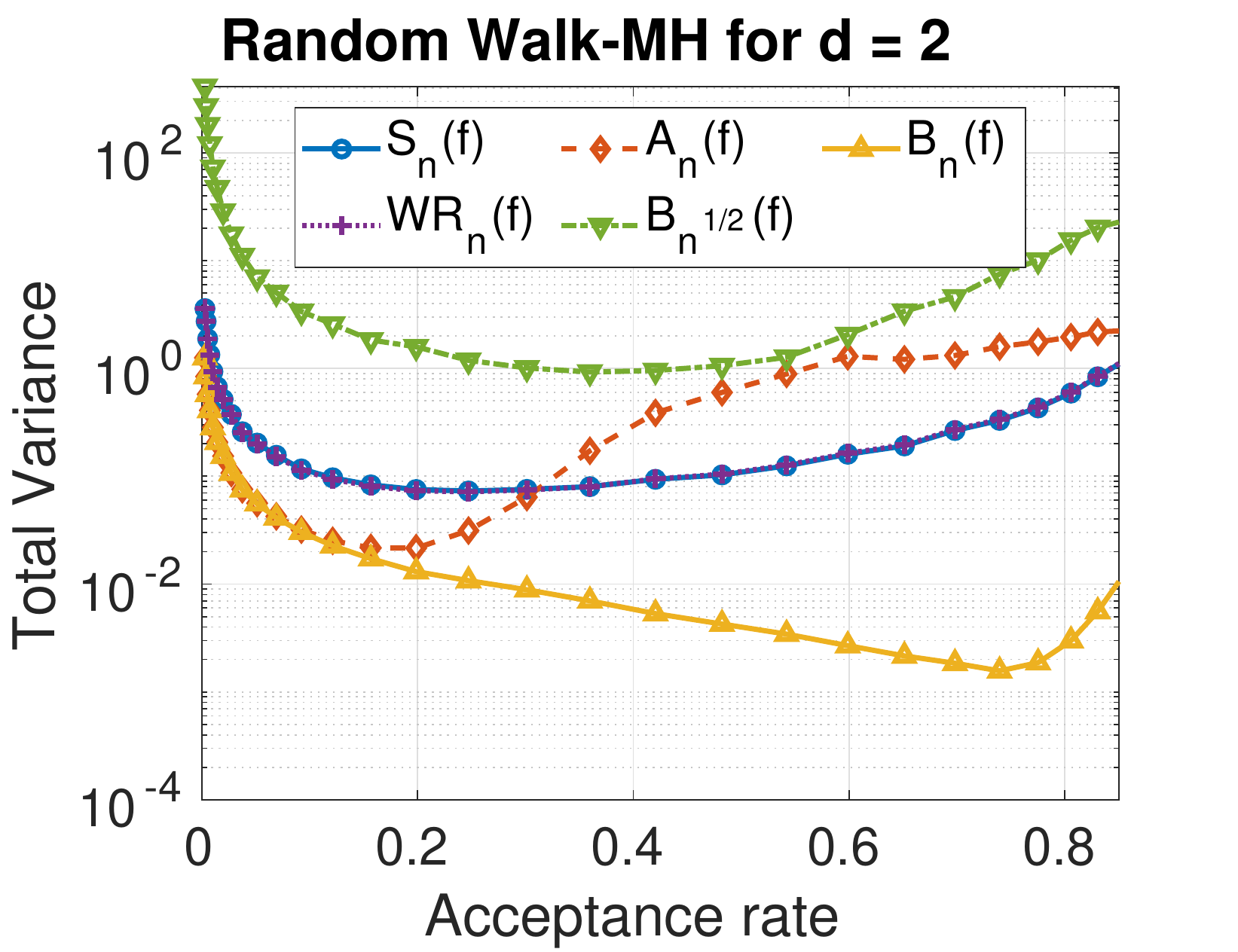}
\includegraphics[width = .48\textwidth]{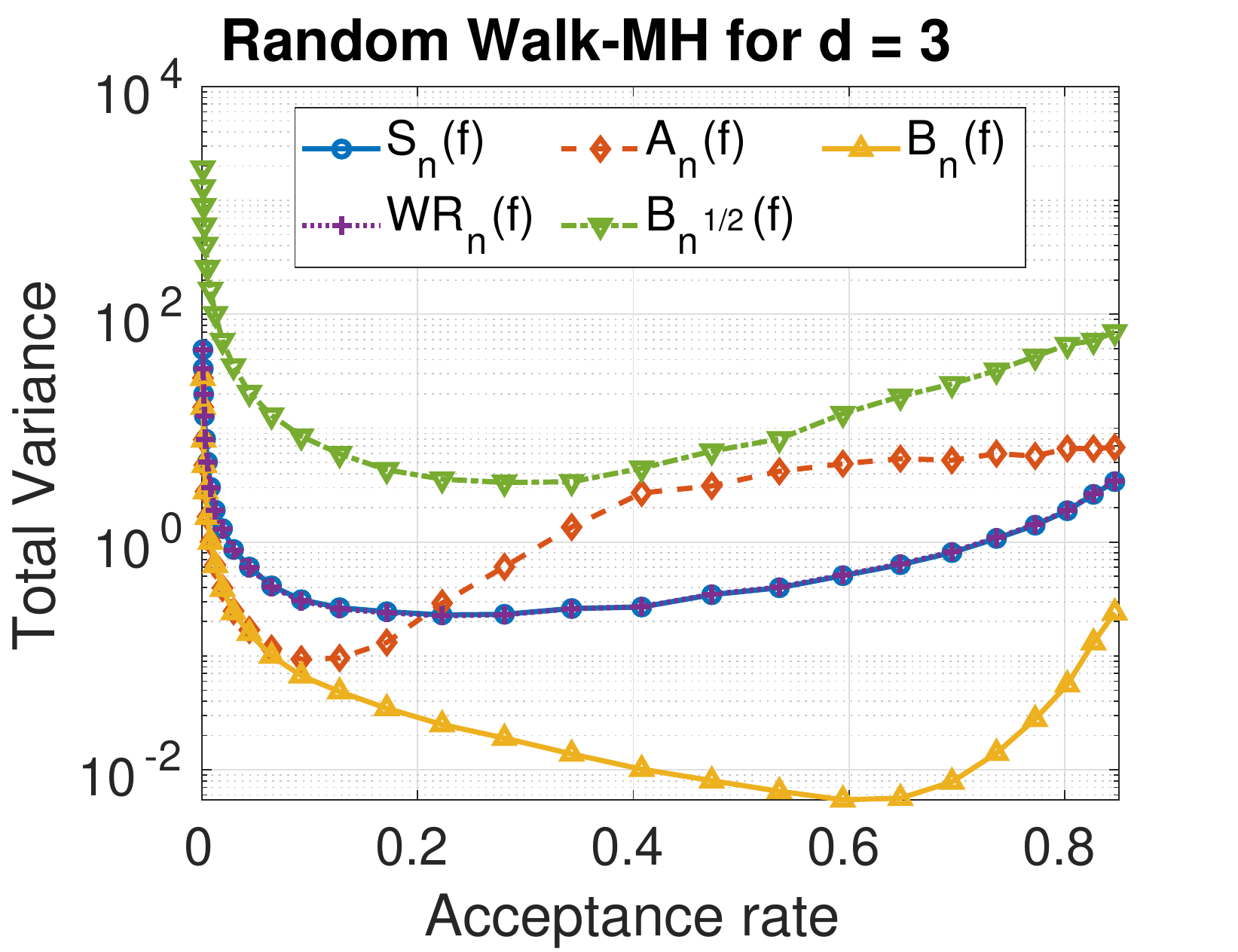}\\[1pt]
\includegraphics[width = .48\textwidth]{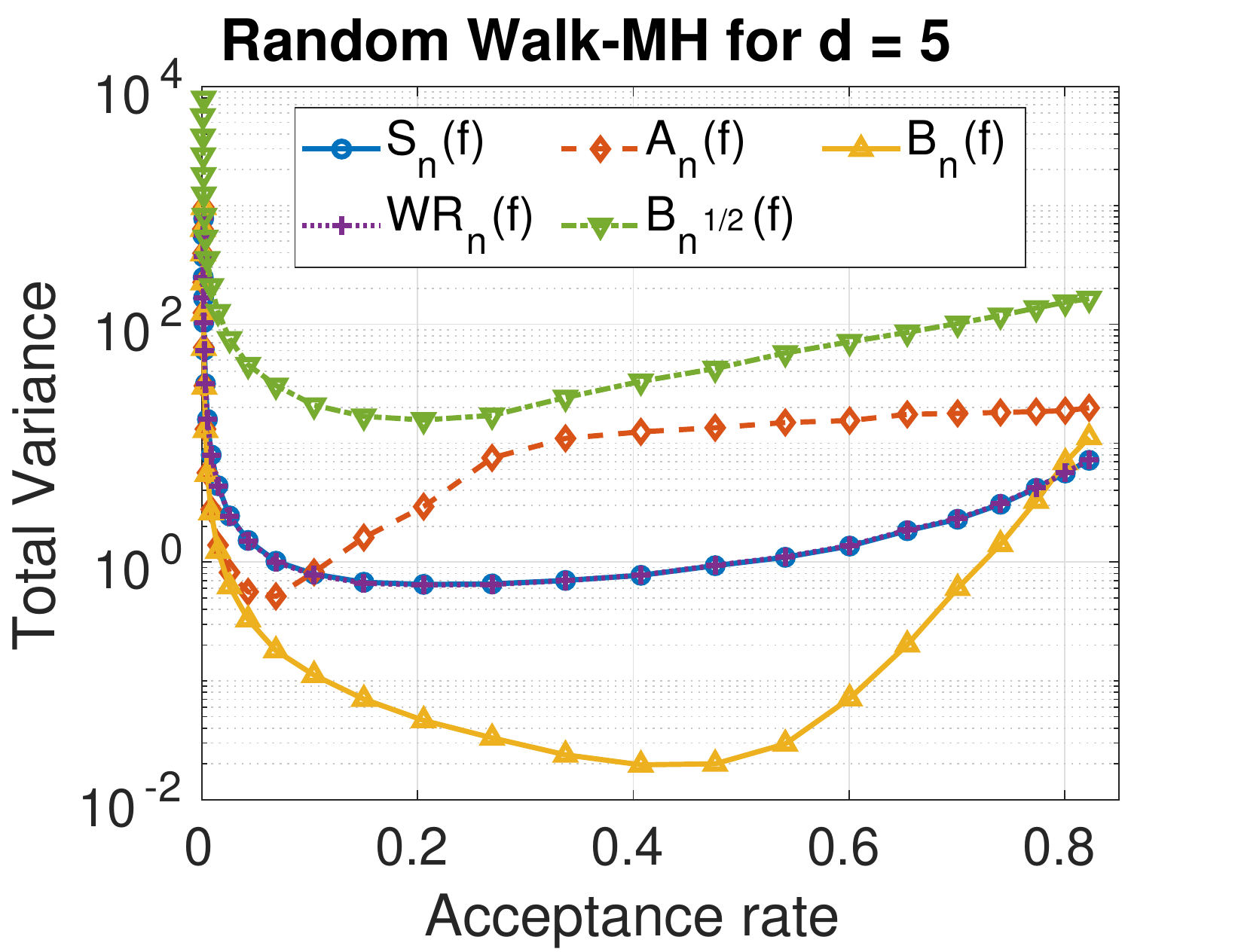}
\includegraphics[width = .48\textwidth]{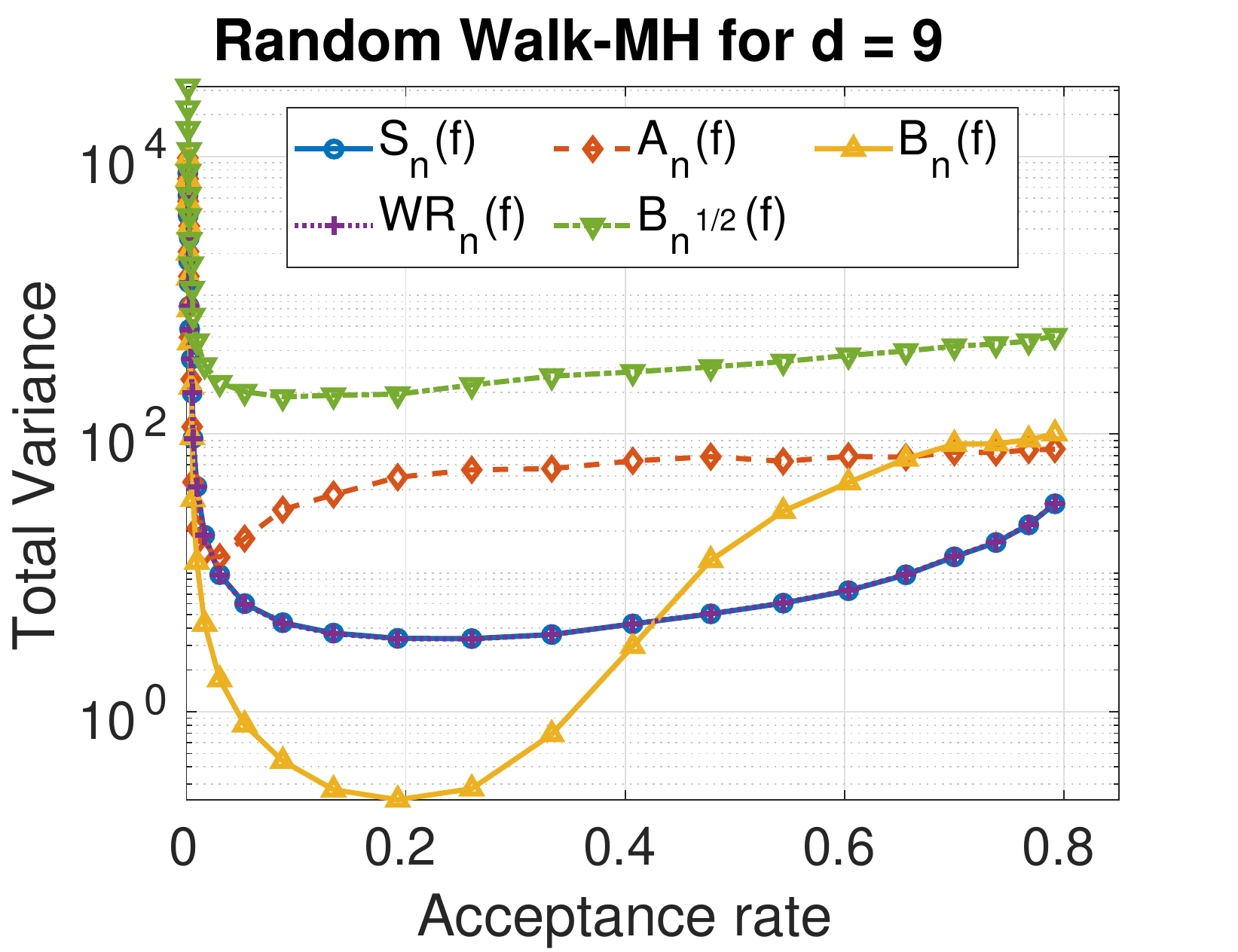}
\caption{Total variances of estimators w.r.t. average acceptance rate in various dimensions for the example from Section \ref{subsec_num_exam_2}.}
\label{fig:exam_results_1_probit}
\end{figure*}

\emph{Performance of $A_n(\mathbf f)$:} For small dimensions, like $d=2,\ldots,5$, we observe that 
the minimal total variance of $A_n(\mathbf f)$ is smaller or at most as large as the minimal total variance of $S_n(\mathbf f)$, see also Table~\ref{tab:rates}. 
However, for dimensions $d\geq6$ the MH importance sampling estimator $A_n(\mathbf f)$ shows a higher total variance than the classical path average estimator $S_n(\mathbf f)$. 
In particular, we observe in Table~\ref{tab:rates} that the performance of $A_n(\mathbf f)$ compared to $S_n(\mathbf f)$ seems to decline more and more for increasing dimension.\par
\emph{Performance of other estimators:} As in Section \ref{subsec_num_exam_1} the waste recycling estimator $WR_n(\mathbf f)$ basically coincides with the path average estimator $S_n(\mathbf f)$ for any considered dimension $d$. 
Also for the Markov chain importance sampling estimators $B_n(\mathbf f)$ and $B_{\sqrt{n}}(\mathbf f)$ the performance compared to $S_n(\mathbf f)$ and $A_n(\mathbf f)$ is similar to Section \ref{subsec_num_exam_1}, i.e., $B_n(\mathbf f)$ outperformes all other estimators --- but at a higher cost --- whereas its cost-equivalent version $B_{\sqrt{n}}(\mathbf f)$ performes worse than any other estimator.
However, also $B_n(\mathbf f)$ and $B_{\sqrt{n}}(\mathbf f)$  seem to suffer from higher dimensions of the state space as indicated in Table~\ref{tab:rates}, i.e., their total variance relative to the total variance of $S_n(\mathbf f)$ becomes larger as $d$ increases.\par
\emph{Optimal calibration:} 
We observe that for $S_n(\mathbf f)$ the total variance becomes minimal for average acceptance rates between $0.2$ and $0.25$.
This is in accordance with the well-known asymptotic result on optimal a-priori step size choices, see \cite{RoRo01}.
The same optimal calibration holds true for the waste recycling estimator $WR_n(\mathbf f)$.
For MH importance sampling estimator $A_n(\mathbf f)$ the minimal total variance is obtained for smaller 
and smaller average acceptance rates as the dimension $d$ increases. 
In fact, the numerical results indicate that the optimal proposal step size $s$ for $A_n(\mathbf f)$ remains constant w.r.t.~the dimension $d$. 
This is in contrast to the classical MCMC estimator where the optimal asymptotic a-priori step size $s$ behaves 
for a product density $\rho$ like $d^{-1}$ for the Gaussian random walk proposal, see \cite{RoRo01}.
Moreover, for each dimension $d$ the minimal total variances of $A_n(\mathbf f)$ where obtained for stepsizes satisfying the optimal calibration rules outlined in Section~\ref{sec:stepsize}.
Concerning the estimators $B_n(\mathbf f)$ and $B_{\sqrt{n}}(\mathbf f)$ we also observe that the optimal performance occurs for decreasing acceptance rates as the dimension $d$ increases.
Here, the numerical results suggest that the optimal proposal stepsize $s$ even increases mildly with the dimensions $d$.

\begin{table}
\caption{Ratio of the total variances $\Var(E_n(\mathbf f))/\Var(S_n(\mathbf f))$ for various estimators $E_n$ and dimensions $d$ for the example of Section~\ref{subsec_num_exam_2}.}
\label{tab:rates} 
\begin{center}
\begin{tabular}{ccccc}
\hline\noalign{\smallskip}
Dimension $d$ & $\frac{\Var({ A}_n({\mathbf f}))}{\Var({S}_n({\mathbf f}))}$ & $\frac{\Var({ B}_n({\mathbf f}))}{\Var({S}_n({\mathbf f}))}$ & $\frac{\Var({WR}_n({\mathbf f}))}{\Var({S}_n({\mathbf f}))}$ & $\frac{\Var({B}_{\sqrt{n}}({\mathbf f}))}{\Var({S}_n({\mathbf f}))}$\\
\noalign{\smallskip}\hline\noalign{\smallskip}
2 & 0.30 & 0.02 & 0.98 & 12.68\\
3 & 0.41 & 0.02 & 0.98 & 14.56\\
4 & 0.59 & 0.03 & 0.99 & 18.63\\
5 & 0.79 & 0.03 & 0.99 & 24.19 \\
6 & 1.40 & 0.03 & 0.99 & 31.98\\
7 & 2.11 & 0.04 & 0.99 & 43.59\\
8 & 2.98 & 0.06 & 0.99 & 48.84\\
9 & 3.86 & 0.07 & 1.00 & 55.10\\
\noalign{\smallskip}\hline
\end{tabular}
\end{center}
\end{table}

\section{Conclusion}\label{sec:concl}
In this work we studied a MH importance sampling estimator 
$A_n$ for which we showed a SLLN, a CLT 
and an explicit estimate of the mean squared error. 
A remarkable property of this estimator is that its 
asymptotic variance does not contain any autocorrelation term, in fact 
\[
 \Corr(\bar \rho(X_k,Y_k) f(Y_k), \bar \rho(X_m,Y_m) f(Y_m)) = \delta_{k}(\{m\}).
\]
This is in sharp contrast to the asymptotic variance of 
the classical MCMC estimator $S_n$, see \eqref{equ:asymp_var_cor}.
Additionally, we performed numerical experiments which indicate that the
MH importance sampling estimator can outperform the classical one.
This requires the correct tuning of the underlying MH Markov chain 
in terms of the proposal step size where the estimator $A_n$ seems to 
benefit from rather small average acceptance rates in contrast to 
optimal scaling results for the MCMC estimator.
However, we exhibit a decreasing efficiency of the MH importance sampling 
estimator for increasing dimension in the numerical experiments.
Indeed, the classical MCMC estimator performs better for larger dimensions.
This is very likely related to the well-known 
degeneration of efficiency for importance sampling in high dimensions, 
see for example the discussion \cite[Section~2.5.4]{AgPaSaSt17}. 

\appendix

\section{Inheritance of Geometric Ergodicity}\label{sec:geomErgodic}
A transition kernel 
$K\colon G\times \mathcal G \to [0,1]$ with stationary distribution 
$\mu$ is \emph{$L^2(\mu)$-geometrically ergodic} 
if there exists a constant $r\in[0,1)$ 
such that for all probability measures $\eta$ on $G$ 
with $\frac{\d \eta}{\d \mu} \in L^2(\mu)$ there is $C_\eta\in[0,\infty)$ satisfying
\begin{equation}\label{eq:geomErgodic}
	d_\text{TV}(\mu, \eta K^n) \leq C_\eta\, r^n \qquad \forall n \in\mathbb N,
\end{equation}
where $d_\text{TV}$ denotes the total variation distance. 
Note that if $\frac{\d \eta}{\d \mu}$ exists, then
\[
	d_\text{TV}(\mu, \eta) := \sup_{A \in \mathcal G} |\mu(A) - \eta(A)| 
	= \frac 12 \int_G \left| \frac{\d \eta}{\d \mu}(x) - 1\right| \mu(\d x).
\]
In addition to the exponential convergence, 
$L^2(\mu)$-geometric ergodicity also yields advantages concerning 
the CLT for the classical MCMC estimator $S_n(f)$ for $\mathbb E_\mu(f)$.
\begin{proposition}[{\cite[
		Corollary 2.1
		]{RoRo97}}]
 Let $(X_n)_{n\in \mathbb{N}}$ be a Markov chain with $\mu$-reversible, $L^2(\mu)$-geometrically ergodic transition kernel.
 Then, for $f\in L^2(\mu)$ 
 we have $\sigma^2_S(f) < \infty$ and 
 $\sqrt{n}(S_n(f)-\mathbb{E}_\mu(f)) \overset{\mathcal{D}}{\longrightarrow}    
  \mathcal{N}(0,\sigma_{S}^2(f))$ as $n\to\infty$.
\end{proposition}
A further important aspect, see e.g. \cite{RoRo97}, 
is the relation between $L^2(\mu)$-geometric ergodicity of a $\mu$-reversible transition kernel $K$ 
and spectral properties of the associated self-adjoint transition operator.
To this end, we introduce $L^2_0(\mu)$ as 
the space of all $g\in L^2(\mu)$ 
satisfying $\mathbb{E}_\mu (g) = 0$.
\begin{proposition}[{\cite[Theorem~2.1]{RoRo97},\cite[Proposition~1.5]{KoMe12}}] \label{th:RR_geomErgodic}
Let the transition kernel $K\colon G\times \mathcal G \to [0,1]$ be $\mu$-reversible.
Then, $K$ is $L^2(\mu)$-geometrically ergodic if and only if
\begin{equation}\label{eq:gap}
	\norm{{\mathrm K}}_{L_0^2(\mu)\to L_0^2(\mu)}
	 < 1.
\end{equation}
\end{proposition}
The condition \eqref{eq:gap} is often referred to as the existence of a positive \emph{$L^2(\mu)$-spectral gap} of $\mathrm K$:
\[
	\mathrm{gap}_\mu(\mathrm K) := 1- \norm{{\mathrm K}}_{L_0^2(\mu)\to L_0^2(\mu)} > 0.
\]
By Lemma~\ref{lem: K_aug_prop} we obtain easily the following relation between the norms of 
$\mathrm K\colon L_0^2(\mu)\to L_0^2(\mu)$ and the corresponding operator 
$\mathrm K_\text{aug}\colon L_0^2(\nu)\to L_0^2(\nu)$ of the augmented MH Markov chain.
\begin{lemma}\label{lem: K_aug_spec}
With the same notation introduced in Section \ref{sec:augMC} we have that
\begin{enumerate}
  \item\label{it: operator_norm_est} 
  $\norm{{\mathrm K}^{n}}_{L_0^2(\mu)\to L_0^2(\mu)} 
  \leq \norm{{\mathrm K}_{\text{aug}}^{n-1}}_{L^2_{0}(\nu)\to L^2_0(\nu)}$
  and $\norm{{\mathrm K}_{\text{aug}}^{n}}_{L^2_{0}(\nu)\to L^2_0(\nu)} 
  \leq \norm{{\mathrm K}^{n-1}}_{L_0^2(\mu)\to L_0^2(\mu)}$ for $n\geq2$;
  \item \label{it: spec_rad_arg}
  $\norm{{\mathrm K}}_{L_0^2(\mu)\to L_0^2(\mu)}
  \leq \norm{{\mathrm K_{\text{aug}}}}_{L^2_{0}(\nu)\to L^2_0(\nu)}$
  and the spectrum of ${\mathrm K_\text{aug}}$ is non-negative and real as well as
  the spectral radius $r({\mathrm K_{\text{aug}}}\mid L^2_0(\nu)) $ of 
  ${\mathrm K}_{\text{aug}}$ on $L^2_0(\nu)$ satisfies
  \[
    r({\mathrm K_{\text{aug}}}\mid L^2_0(\nu)) \leq \norm{{\mathrm K}}_{L_0^2(\mu)\to L_0^2(\mu)}.
  \]
\end{enumerate}
\end{lemma}
\begin{proof}
 \textbf{To \ref{it: operator_norm_est}.:}
Note that ${\mathrm {\widehat P}^*}f \in L^2_0(\nu)$, ${\mathrm H} g\in L^2(\nu)$ 
and ${\mathrm{\widehat P}}g\in L_0^2(\mu)$
for any $f\in L_0^2(\mu)$ and $g\in L^2_0(\nu)$.
By applying Lemma~\ref{lem: K_aug_prop} we have
\begin{align*}
 \norm{{\mathrm K}^n}_{L_0^2(\mu)\to L_0^2(\mu)} 
 & = \norm{\mathrm{\widehat P K}_{\text{aug}}^{n-1}\mathrm{ H \widehat P^*}}
 _{L_0^2(\mu)\to L_0^2(\mu)} \leq \norm{{\mathrm K}_{\text{aug}}^{n-1}}_{L^2_{0}(\nu) \to L^2_{0}(\nu)},
\end{align*}
since 
\[
\norm{{\mathrm{\widehat P}}}_{L_0^2(\nu)\to L^2_0(\mu)}
\leq \norm{{\mathrm{\widehat P}}}_{L^2(\nu)\to L^2(\mu)}=1
\quad
\text{ and }
\quad
\norm{\mathrm{ H \widehat P}^*}_{L^2_0(\mu)\to L^2_0(\nu)}\leq 1.
\]
Similarly
\begin{align*}
 \norm{{\mathrm K}_{\text{aug}}^n}_{L^2_0(\nu)\to L^2_0(\nu)} 
 & = \norm{\mathrm{ H \widehat P^* K}^{n-1}\mathrm{\widehat P}} _{L^2_0(\nu)\to L^2_0(\nu)}
  \leq \norm{{\mathrm K}^{n-1}}_{L^2_0(\mu)\to L^2_0(\mu)}.
\end{align*}
\textbf{To \ref{it: spec_rad_arg}.:}
By the fact that ${\mathrm K} \colon L^2_0(\mu)\to L^2_0(\mu)$ 
is self-adjoint, properties of the spectral radius formula for self-adjoint operators 
and 
statement \ref{it: operator_norm_est}~we have
\begin{align*}
 \norm{{\mathrm K}}_{L^2_0(\mu)\to L^2_0(\mu)}
& 
  = \lim_{n\to \infty} (\norm{{\mathrm K}^n}_{L^2_0(\mu)\to L^2_0(\mu)})^{1/n}
 \leq  \lim_{n\to \infty} (\norm{{\mathrm K}_{\text{aug}}^{n-1}}_{L_0^2(\nu)\to L_0^2(\nu)})^{1/n}\\
 & = \norm{{\mathrm K}_{\text{aug}}}_{L^2_0(\nu)\to L^2_0(\nu)}.
\end{align*}
Unfortunately, $K_{\text{aug}}$ is in general not reversible, see Remark~\ref{rem: K_aug_not_rev}, such that ${\mathrm K_\text{aug}}$ is not self-adjoint.
Thus, we can only 
estimate the spectral radius of ${\mathrm K}_{\text{aug}}\colon L_0^2(\nu) \to L_0^2(\nu)$, 
but not the operator norm. The same argument yields to
\[
r({\mathrm K_{\text{aug}}}\mid L_0^2(\nu)) \leq \norm{{\mathrm K}}_{L^2_0(\mu)\to L^2_0(\mu)}.
\]
Finally, since ${\mathrm K}_{\text{aug}}$ is a product 
of two self-adjoint operators and, additionally, 
the projection $\mathrm{\widehat P^* \widehat P}$ is positive, 
we obtain by \cite[Proposition~4.1]{RoRo15} that the spectrum of 
${\mathrm K}_{\text{aug}}\colon L_0^2(\nu) \to L_0^2(\nu)$ is real and non-negative.
\end{proof}
Since $K_\text{aug}$ is not reversible, we can not argue that a positive $L^2(\nu)$-spectral 
gap of $\mathrm K_\text{aug}$, which due to statement \ref{it: spec_rad_arg} of Lemma \ref{lem: K_aug_spec} is implied by a positive $L^2(\mu)$-spectral gap of $\mathrm K$, yields the $L^2(\nu)$-geometric ergodicity of the augmented MH Markov chain.
However, by using also statement \ref{it: operator_norm_est} of Lemma \ref{lem: K_aug_spec} we indeed obtain the inheritance of geometric ergodicity.
\begin{corollary}
Assume that the MH transition kernel $K$ with stationary distribution $\mu$ on $G$
is $L^2(\mu)$-geometrically ergodic.
Then, the augmented MH transition kernel ${\mathrm K}_{\text{aug}}$ 
is $L^2(\nu)$-geometrically ergodic with $\nu$ as in \eqref{equ:nu}.
\end{corollary}
\begin{proof}
By Proposition~\ref{th:RR_geomErgodic} and the $\mu$-reversibility of $K$ we have 
that $r:= \norm{{\mathrm K}}_{L_0^2(\mu)\to L_0^2(\mu)} < 1$.
Let $\eta$ be a probability distribution on $G\times G$ 
such that $\frac{\d \eta}{\d \nu} \in L^2(\nu)$. With the notation of the adjoint operator
we use (for details we refer to \cite[Lemma~3.9]{Ru12}) that
\[
 \frac{\d (\eta K^n_\text{aug})}{\d \nu}(x,y) 
 = (\mathrm K^n_\text{aug})^*\left[\frac{\d \eta}{\d \nu}\right](x,y),\quad \nu\text{-a.e.}
\]
as well as
\[
	\norm{(\mathrm K^n_\text{aug})^*}_{ L^2_0(\nu)\to L^2_0(\nu)}=\norm{\mathrm K^n_\text{aug}}_{ L^2_0(\nu)\to L^2_0(\nu)}.
\]
Then, for $n\geq 2$ we have
\begin{align*}
2 d_\text{TV}(\nu, \eta K_\text{aug}^n)
& = \int_{G\times G} \left| \frac{\d (\eta K^n_\text{aug})}{\d \nu}(x,y) - 1\right| \nu(\d x \d y)\\
& = \int_{G\times G} \left| (\mathrm K^n_\text{aug})^*\left[\frac{\d \eta}{\d \nu}\right](x,y) - 1\right| \nu(\d x\d y)\\
& = \int_{G\times G} \left| (\mathrm K^n_\text{aug})^*\left[\frac{\d \eta}{\d \nu}(x,y) - 1\right]\right| \nu(\d x \d y)\\
& \leq \norm{(\mathrm K^n_\text{aug})^*\left[\frac{\d \eta}{\d \nu}- 1\right]}_{\nu}\\
& \leq \norm{\mathrm K^n_\text{aug}}_{ L^2_0(\nu)\to L^2_0(\nu)} \norm{\frac{\d \eta}{\d \nu}- 1}_{\nu}\\
&\leq \norm{\mathrm K}^{n-1}_{ L^2_0(\mu)\to L^2_0(\mu)} \norm{\frac{\d \eta}{\d \nu} - 1}_{\nu}
\leq C_\eta\ r^n
\end{align*}
with $C_\eta := \frac1r \norm{\frac{\d \eta}{\d \nu} - 1}_{\nu}$, where we used 
the fact that $(\frac{\d \eta}{\d \nu}- 1) \in L^2_0(\nu)$ as well as 
statement \ref{it: operator_norm_est} of Lemma~\ref{lem: K_aug_spec}.
\end{proof}

\paragraph{Acknowledgements:} 
We thank Ilja Klebanov,  Krzysztof {\L}atuszy{\'n}ski, Anthony Lee,
Ingmar Schuster, and Matti Vihola for helpful and inspiring discussions.
Daniel Rudolf gratefully acknowledges support of the 
Felix-Bernstein-Institute for Mathematical Statistics in the Biosciences 
(Volkswagen Foundation) and the Campus laboratory AIMS.
Bj\"orn Sprungk is grateful for the financial support by the DFG within the project 389483880.


\bibliographystyle{amsalpha}
\bibliography{lit}  

%
%
%
%
%

\end{document}